\documentclass[11pt,letterpaper]{article}
\usepackage{fullpage,amsmath,amsthm,amsfonts,amssymb,commath,xcolor}
\usepackage[shortlabels]{enumitem}
\usepackage[ruled,vlined,linesnumbered]{algorithm2e}
\usepackage{amsmath,amsfonts,amsthm,amssymb}
\usepackage{graphicx,float,wrapfig}
\usepackage{placeins}
\usepackage{algpseudocode}

\usepackage{tikz}
\usetikzlibrary{shapes,positioning,arrows.meta,fit}
\usepackage{wrapfig}
\usepackage{caption, subcaption}

\theoremstyle{definition}

\newtheorem{theorem}{Theorem}[section]
\newtheorem{corollary}[theorem]{Corollary}
\newtheorem{lemma}[theorem]{Lemma}

\newtheorem{definition}[theorem]{Definition}

\newtheorem*{remark}{Remark}

\newcommand{\todo}[1]{{\color{red} todo: #1}}

\newcommand{\Brac}[1]{\left[#1\right]}
\newcommand{\Ex}[1]{\mathbb{E}\Brac{#1}}
\newcommand{\Var}[1]{\mathrm{Var}\Brac{#1}}
\newcommand{\poly}{\mathrm{poly}}
\newcommand{\sk}{\mathrm{sk}}

\newcommand{\dun}[2]{\delta_{#1}^{\mathrm{un}}(#2)}

\newcommand{\calG}{\mathcal{G}}

\newcommand{\eps}{\epsilon}
\renewcommand{\tilde}{\widetilde}

\newcommand{\eat}[1]{}

\newcommand*\samethanks[1][\value{footnote}]{\footnotemark[#1]}

\title{Sparsification of Directed Graphs via Cut Balance}

\author{
\begin{tabular}{c c}
  \begin{tabular}{c}
Ruoxu Cen\,\\Tsinghua University
  \end{tabular} &
  \begin{tabular}{c}
Yu Cheng\,\thanks{
Part of this work was done when the author was a postdoctoral researcher at Duke University and when he was visiting the Institute of Advanced Study.}\\University of Illinois at Chicago
  \end{tabular} \\ \\
  \begin{tabular}{c}
Debmalya Panigrahi\,\thanks{This work was supported in part by NSF grants CCF-1535972, CCF-1955703, and an NSF CAREER Award CCF-1750140.}\\Duke University
  \end{tabular} &
  \begin{tabular}{c}
Kevin Sun\,\samethanks\\Duke University
  \end{tabular}
\end{tabular}
}
\date{}

\begin{document}

\maketitle
\setcounter{page}{0}
\thispagestyle{empty}

\begin{abstract}
    In this paper, we consider the problem of designing cut sparsifiers and sketches for directed graphs. To bypass known lower bounds, we allow the sparsifier/sketch to depend on the {\em balance} of the input graph, which smoothly interpolates between undirected and directed graphs. We give nearly matching upper and lower bounds for both {\em for-all} (cf.~Bencz\'ur and Karger, STOC 1996) and {\em for-each} (Andoni {\em et al.}, ITCS 2016) cut sparsifiers/sketches as a function of cut balance, defined the maximum ratio of the cut value in the two directions of a directed graph (Ene {\em et al.}, STOC 2016). We also show an interesting application of digraph sparsification via cut balance by using it to give a very short proof of a celebrated maximum flow result of Karger and Levine (STOC 2002).
\end{abstract}

\clearpage

\section{Introduction}
Graph sparsification, originally introduced by Bencz\'ur and Karger as a means of obtaining faster maximum flow algorithms~\cite{benczur2015randomized}, has become a fundamental tool in graph algorithms. The goal of graph sparsification is to replace an arbitrary graph with a sparse graph (called the graph sparsifier) on the same set of $n$ vertices but with only $O(n \cdot \poly(\log n, 1/\epsilon))$ edges, while approximately preserving the value of every cut up to a factor of $1\pm \epsilon$ for any given $\epsilon > 0$.

Since their work, several graph sparsification techniques have been discovered~(e.g., \cite{fung2011general}), the idea has been extended to other models of computation such as data streaming (e.g.,~\cite{ahn2009graph}) and sketching (e.g.,~\cite{andoni2016sketching}), stronger notions such as spectral sparsification that preserves all quadratic forms have been proposed (e.g.,~\cite{SpielmanT11}), and far-reaching generalizations such as the Kadison-Singer conjecture have been established~\cite{marcus2015interlacing}. On the applications front, graph sparsification has been heavily used to obtain a tradeoff between algorithmic accuracy and efficiency for a variety of ``cut-based'' problems such as maximum flows, minimum cuts, balanced separators, etc.

\begin{wrapfigure}{r}{0.28\textwidth}
\centering
\includegraphics[width=0.26\textwidth]{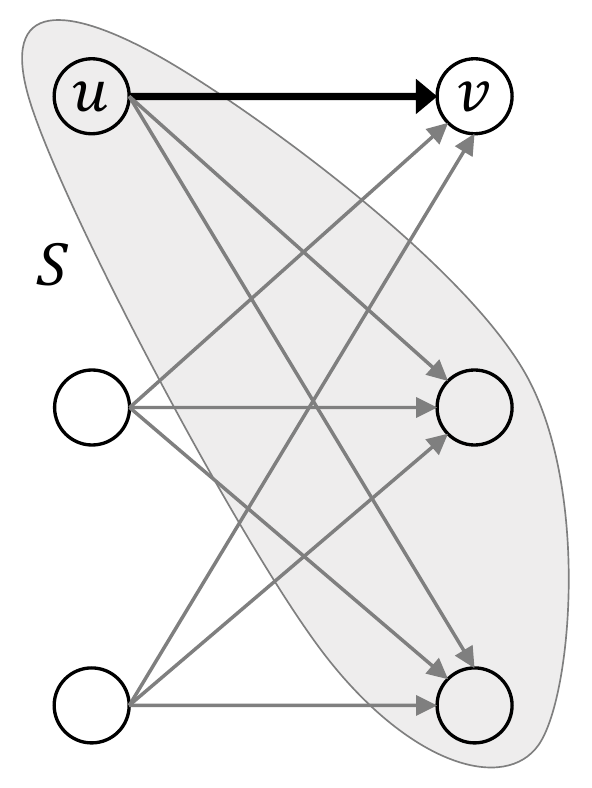}
\caption{The cut $S$ only contains a single edge $(u,v)$, so any sparsifier must contain this edge.}
\label{fig:complete}
\vspace{-5mm}
\end{wrapfigure}

In spite of its widespread use, one restriction is that most sparsification techniques only apply to undirected graphs. There is a fundamental reason for this restriction -- there are directed graphs that {\em cannot} be sparsified (see Fig.~\ref{fig:complete} for an example). 

Indeed, the lower bound holds even for {\em cut sketches}, where one does not insist on a graph being output as the sparsifier, but simply a succinct data structure from which the cut values of the original graph can be (approximately) retrieved.

A qualitative distinction between directed and undirected graphs is in terms of the {\em balance} of cuts, i.e., the ratio between incoming and outgoing edges in any given cut. An equivalent view of an undirected graph is by bi-directing its edges, which results in a graph with perfect balance, i.e., every cut has exactly the same number of incoming and outgoing edges.\footnote{Note that all Eulerian digraphs, whether or not derived from undirected graphs, exhibit perfect cut balance.} Cut balance, therefore, smoothly interpolates between undirected and directed graphs, which leads to the question: {\em can we design cut sparsifiers/sketches for directed graphs that depend on cut balance?} We answer this question in the affirmative in this paper, and show that this view of sparsification leads to interesting consequences.


We note that the use of cut balance to bridge between undirected and directed graphs predates our work. Ene {\em et al.}~\cite{ene2016routing} introduced the notion of parameterizing digraphs by their cut balance (or simply balance) and defined it as the largest ratio between the value of a cut in its two directions. Using this view, they extended two classic operations on undirected graphs -- oblivious routing and fast approximate maximum flows -- to directed graphs with a dependence on the balance. In this paper, we show that this phenomenon is exhibited by cut sparsification as well. 

\subsection{Our Results}

We consider the two canonical forms of cut sparsification considered in the literature. The first is the classic version introduced by  Bencz\'ur and Karger~\cite{benczur2015randomized}, where all cuts must simultaneously be approximately preserved {\em whp};\footnote{with high probability} we call this {\em for-all} sparsification. The second, more relaxed, notion is due to Andoni {\em et al.}~\cite{andoni2016sketching}, where {\em any} cut must be approximately preserved {\em whp} instead of all cuts simultaneously; we call this {\em for-each} sparsification. For both these notions of sparsification, previous results on undirected graphs can be extended to $\beta$-balanced graphs by boosting sampling probabilities in undirected sparsification algorithms by a factor of $\beta$, thereby losing an additional factor of $\beta$ in the size of the sparsifier/sketch (see also Ikeda and Tanigawa~\cite{ikeda2018cut}). Is it possible to do better than losing a factor of $\beta$? 

Our first result sharpens this na\"ive bound in for-each sparsification, by constructing cut sketches that improve the dependence on $\beta$ to $\sqrt{\beta}$. We also show that this dependence is tight by constructing a matching lower bound. This pair of results resolves the precise dependence of for-each sparsification in directed graphs on the balance of the graph.

\begin{theorem}
\label{thm:foreach-intro}
(Upper Bound) For any $\beta$-balanced graph with $n$ vertices, $m$ edges, and polyno\-mially-bounded edge weights, there is an $\tilde O(m + \sqrt{\beta} n / \eps)$-time algorithm\footnote{This runtime bound assumes that the value of $\beta$ is known. If not, then $\beta$ can be computed using an algorithm of Ene~{\em et al.}~\cite{ene2016routing} in $\tilde O(\beta^2 m)$ time.} that constructs a $(1 \pm \eps)$ for-each cut sketch of size $\tilde O(\sqrt{\beta} n / \eps)$ bits.

\smallskip\noindent
(Lower Bound) Fix any $\beta \ge 1$, $0 < \eps < 1$, and $n$ such that $(\beta/\eps)^{1/2} \le n/2$. Any $(1\pm\eps)$ for-each cut sketching algorithm for $\beta$-balanced graphs with $n$ vertices must output at least $\Omega(\sqrt{\beta} n / \sqrt{\eps})$ bits in the worst case.
\end{theorem}

In for-all sparsification, we are not as lucky; we show that the linear dependence on $\beta$ is tight in this case. (In fact, Ikeda and Tanigawa~\cite{ikeda2018cut} had conjectured that better for-all sparsifiers can be constructed by sampling edges according to directed connectivity parameters; our lower bound construction refutes this conjecture and shows that such more aggressive sampling may not produce a sparsifier at all.)

\begin{theorem}
\label{thm:forall-intro}
Fix any $\beta \ge 1$, $0 < \eps < 1$, and $n$ such that $\beta/\eps \le n/2$.
Any $(1\pm\eps)$ for-all cut sketching algorithm for $n$-node $\beta$-balanced graphs must output at least $\Omega(n \beta / \eps)$ bits. 
\end{theorem}

But, we note that the upper bound only applies to digraphs where {\em all} cuts are $\beta$-balanced. In general, the balance parameter for different cuts in a digraph may be highly non-uniform: some cuts could be very balanced and some others very unbalanced. For such graphs, we show a more refined result: for {\em any} value $\beta \ge 1$, we construct a sparsfier that approximately preserves all $\beta$-balanced cuts losing only an additional factor $\beta$ in the size of the sparsifier. Note that this result holds for any value of $\beta$ irrespective of the balance parameter of the graph; if $\beta$ is the balance parameter, then it recovers the tight bound for $\beta$-balanced graphs.

\begin{theorem}
\label{thm:forall-strong}
For any directed graph with $n$ vertices, $m$ edges, and non-negative edge weights, and any $\beta\ge 1$, there is an $\tilde O(m)$-time algorithm that returns a (weighted) subgraph with $\tilde O(\beta n/\epsilon^2)$ edges and preserves the values of all $\beta$-balanced cuts up to a factor of $1\pm\eps$.
\end{theorem}

We remark that digraph sparsification using cut balance has interesting consequences. In particular, note that for residual graphs produced by $s$-$t$ maximum flow algorithms in undirected graphs, we can precisely bound the balance parameter on all cuts separating the source and the sink. Using this observation, we give a very short proof of the celebrated maximum flow result of Karger and Levine~\cite{karger2002random} via the digraph sparsification results.


\subsection{Our Techniques}
First, we outline the main ideas in our for-each cut sketch. In previous results on for-each cut sketches of undirected graphs~\cite{andoni2016sketching,JambulapatiS18}, the main idea was to (recursively) partition the graph into ``sparse'' and ``dense'' parts, and then maintain the sparse parts exactly along with a sample of the dense parts. A directed subgraph, however, can simultaneously be too dense to preserve exactly but also not amenable to sampling (e.g., a complete bipartite digraph). Of course, the balance parameter helps bridge this gap, but the cut balance of a subgraph that the algorithm encounters during recursion can be much worse than that of the original graph. Indeed, individual subgraphs might not even be strongly connected (i.e., have balance $\infty$), even if the original graph were Eulerian (i.e., has balance $1$). This makes the (recursive) local sketching techniques in previous works unusable for directed graphs.

Our main technical contribution is a new {\em global} cut sketch construction. We design a cut sketch whose variance can be large on individual dense regions of the input digraph that are well-connected in an undirected sense, but we crucially show that the {\em cumulative variance of our estimator across all these well-connected regions of the digraph is small}. This helps eliminate the need for local cut sketches in each dense subgraph, and simplifies the recovery algorithm to the natural estimator that appropriately scales the number of sampled edges in the queried cut.
Moreover, to obtain the right dependence on $\beta$, we need to carefully analyze the variance of our estimator. Our new variance analysis works for undirected graphs as well, which tightens the analysis of~\cite{andoni2016sketching} and consequently leads to undirected cut sketching algorithms that do not require downsampling or low-accuracy for-all sparsifiers.


Next, we turn to for-all sparsification. Our first result is the lower bound on for-all sparsifiers and cut sketches. For any $\beta$ and $n$, we construct a family $\cal G$ of $\beta$-balanced graphs on $n$ vertices that satisfies two conflicting properties: $\cal G$ is a large family, yet for each graph $G \in \cal G$, the number of graphs in $\cal G$ that approximate all cuts of $G$ is small. For any graph $G$ and cut $S$, let $\delta_G(S)$ denote the value of $S$ and $E_G(S)$ denote the edges crossing $S$. Notice that there are many possible graphs $H$ such that $\delta_H(S) \approx \delta_G(S)$, because $E_H(S)$ and $E_G(S)$ could differ in numerous ways. Thus, to ensure that the number of cut approximators is small, we carefully design $\cal G$ such that for any $G, H \in \cal G$ and cut $S$, if $\delta_H(S) \approx \delta_G(S)$, then $E_H(S) \approx E_G(S)$. We show that this can be done by considering a large family of bipartite graphs that all contain a fixed (directed) matching. Consequently, any sketching algorithm must produce a large number of different cut sketches for the graphs in $\cal G$, which translates to a lower bound on the size of the cut sketches using standard information theory.

Finally, we refine for-all sparsification in digraphs by showing that we can preserve all balanced cuts, irrespective of the balance parameter of the entire graph $G$.
More specifically, at a cost of an additional factor of $\beta$ in the size of the sparsifer, we can preserve all $\beta$-balanced cuts, and provide an approximation for $\alpha$-balanced cuts with $\alpha > \beta$ that degrades gracefully as $\alpha$ gets larger.
For this purpose, we adopt a (recursive) graph decomposition due to Bencz\'ur and Karger~\cite{benczur2015randomized} that expresses a graph as a weighted sum of subgraphs, each of which corresponds to a particular edge sampling rate.
Now we can boost the (undirected) sampling rate by a factor of $\beta$.
If the balance of every subgraph in the decomposition is also $\beta$, then the undirected analysis carries over to the directed case.
However, in general, each subgraph can be very unbalanced, so we cannot bound the estimation error in each individual subgraph. Our main technical contribution is to show that even though we do not preserve the cut values in individual subgraphs, we do so globally across all the subgraphs.

\subsection{Related Work}

{\bf Graph Sparsification.}
Graph sparsification was introduced by Bencz\'ur and Karger~\cite{benczur2015randomized}
(``for-all'' cut sparsification), and has led to research in a number of directions:
Fung~{\em et al.}~\cite{fung2011general} and Kapralov and Panigrahy~\cite{KapralovP12}
gave new algorithms for preserving cuts in a sparsifier; 
Spielman and Teng~\cite{SpielmanT11} generalized to spectral sparsfiers that preserved
all quadratic forms, which led to further research both in reducing the size of the sparsifier~\cite{spielman2011graph,batson2012twice} and developing faster algorithms (e.g., \cite{LeeS15, AllenLO15,LeeS17,ChuGPSSW18,KyngLPSS16,KoutisLP12,KoutisMP11,KoutisMP10}); 
faster algorithms for fundamental graph problems such as maximum flow 
utilized sparsification results (e.g.,~\cite{benczur2015randomized,Sherman13});
Ahn and Guha~\cite{ahn2009graph} introduced sparsification in the streaming model,
which has led to a large body of work for both cut (e.g.,~\cite{AhnGM12a,AhnGM12b,GoelKK10})
and spectral sparsifiers (e.g., \cite{KapralovNST19,KapralovMMMN19,KapralovLMMS17,AhnGM13}) in graph streams; 
both cut~\cite{KoganK15,NewmanR13} and spectral~\cite{SomaY19} sparsification have been
studied in hypergraphs.
%
%
%
For lower bounds, Andoni~{\em et al.}~\cite{andoni2016sketching} showed that any data structure that $(1\pm\eps)$-approximately stores the sizes of all cuts in an undirected graph must use $\Omega(n/\eps^2)$ bits.
Carlson {\em et al.}~\cite{CarlsonKST19} improved this lower bound to $\Omega(n \log n/\eps^2)$ bits, matching existing upper bounds.

Andoni {\em et al.}~\cite{andoni2016sketching} first proposed the notion of ``for-each'' cut (and spectral) sketches, where the sparsifier preserves the value of {\em any} cut rather than all cuts simultaneously.
They showed that for any undirected graph with $n$ vertices, a $(1\pm\eps)$ for-each cut sketch of size $\tilde O(n/\eps)$ exists and can be computed in polynomial time.
Subsequently, Jambulapati and Sidford~\cite{JambulapatiS18} gave the first nearly-linear time algorithm for constructing $(1\pm\eps)$ for-each graph sketches of size $\tilde O(n/\eps)$. Their sketch not only approximates cut values, but also approximately preserves the quadratic form of any undirected Laplacian matrix (and its pseudoinverse). Chu {\em et al.}~\cite{ChuGPSSW18} showed how to construct a {\em graph} containing $\tilde O(n^{1+o(1)}/\eps)$ edges that satisfies the ``for-each'' requirement for spectral queries.

\medskip
{\noindent \bf Directed Graphs.}
Cohen {\em et al.}~\cite{CohenKPPRSV17,CohenKPPRSV18} proposed a directed notion of spectral sparsifiers 
and used it to obtain nearly-linear time algorithms for solving directed Laplacian linear systems and 
computing various properties of directed random walks.
However, their directed spectral sparsifiers only work for Eulerian graphs, i.e., for $\beta = 1$.
Zhang {\em et al.}~\cite{ZhangZF19} proposed a notion of spectral sparsification that works for all directed graphs, 
but their definition does not preserve  cut values. More generally, there have been attempts at bridging the divide 
between directed and undirected graphs for other problems. For instance,
Lin~\cite{lin2009reducing} defined the imbalance of a graph as the sum of the absolute difference of 
in- and out-capacities at all vertices, and used it to generalize
the max-flow algorithm of Karger and Levine~\cite{karger2002random} from undirected graphs to digraphs. 
Digraphs have also been parameterized by directed extensions of treewidth~\cite{johnson2001directed},
and similar notions of DAG-width~\cite{berwanger2012dag,obdrvzalek2006dag} and 
Kelly-width~\cite{hunter2008digraph}, which led to FPT algorithms based on
these parameters, much like for undirected bounded treewidth graphs.
In spectral graph theory, directed analogs of Cheeger's inequality have been defined~\cite{chung2005laplacians},
particularly in the context of analyzing the spectrum of digraphs.
Closest to our work is that of Ene~{\em et al.}~\cite{ene2016routing} who proposed cut balance 
of digraphs that we use in this paper, although in the context of oblivious routing and max-flow algorithms.

\section{Preliminaries}
{\noindent \bf Basic Notations.}
Let $G = (V, E, w)$ be a weighted \emph{directed} graph with $n = |V|$ vertices and $m = |E|$ edges.
Every edge $e \in E$ has a given non-negative weight $w_e\geq 0$.
When working with unweighted graphs, (i.e., $w_e = 1$ for all $e \in E$), we will omit the edge weights $w_e$.

For two sets of vertices $S \subseteq V$ and $T \subseteq V$, we use $E(S, T) = \{(u, v) \in E: u \in S, v \in T\}$ to denote the set of edges in $E$ that go from $S$ to $T$.
We use $w(S, T) = \sum_{e \in E(S,T)} w_e$ to denote the total weight of the edges from $S$ to $T$.
For a vertex $u \in V$ and a set of vertices $S \subseteq V$, we write $E(u, S)$ for $E(\{u\}, S)$, and we define $E(S, u)$, $w(u, S)$, and $w(S, u)$ similarly.

We often write $\overline{S}$ as a shorthand for $V \setminus S$.
Given a component $V_i$ and a subset of its vertices $S_i \subseteq V_i$, we can similarly define $\overline{S_i} = V_i \setminus S_i$.
For example, using this notation, we write $w(S,\overline{S})$ for $w(S, V \setminus S)$ and similarly $w(S_i, \overline{S_i}) = w(S_i, V_i \setminus S_i)$.

The {\em conductance} of an undirected graph $G = (V, E, w)$ is defined as
\begin{align}
\phi(G) &= \min_{\varnothing \neq S \subset V} \frac{w(S, \overline S)}{\min\bigl(w(S, V), w(\overline S, V) \bigr)}.
\label{eqn:conductance}
\end{align}

%
\begin{definition}[$\beta$-Balanced]
A strongly connected digraph $G = (V, E, w)$ is \emph{$\beta$-balanced} if, for all $\varnothing \subseteq S \subseteq V$, 
it holds that $w(S, \overline{S}) \le \beta \cdot w(\overline{S}, S)$.
\end{definition}

{\noindent \bf Directed Cut Sparsifiers and Cut Sketches.}
We consider two notions of sparsification. The first is the classic ``for-all'' sparsifier
that approximately preserves the values of all cuts.


\begin{definition}[For-All Cut Sparsifier]
Let $G = (V, E_G, w_G)$ and $H = (V, E_H, w_H)$ be two weighted directed graphs.
Fix $0 < \eps < 1$.
We say $H$ is a $(1 \pm \eps)$ for-all cut sparsifier of $G$ iff the following holds for all $S \subseteq V$:

\medskip
$
(1-\eps) \cdot w_G(S, V \setminus S) \le w_H(S, V \setminus S) \le (1+\eps) \cdot w_G(S, V \setminus S).
$
\end{definition}

Instead of a graph that preserves cut values, if we allow any data structure from which the cut values can be (approximately) recovered, 
we call it a cut sketch.

\begin{definition}[For-All Cut Sketch]
Let $G = (V, E, w)$ be a weighted directed graph.
Fix $0 < \eps < 1$.
A (deterministic) function $g$ outputs a $(1\pm\eps)$ for-all cut sketch of $G$ if there exists a recovering function $f$ such that, for all $S \subseteq V$:

\medskip
$
(1-\eps) \cdot w(S, V \setminus S) \le f(S, \sk(G)) \le (1+\eps) \cdot w(S, V \setminus S).
$
\end{definition}

Next we consider a weaker notion of graph sparsification, where instead of approximating the value of all cuts, we only require the value of any individual cut to be approximately preserved with (high) constant probability.

\begin{definition}[For-Each Cut Sketch]
Let $G = (V, E, w)$ be a weighted directed graph.
Fix $0 < \eps < 1$.
A function $g$ outputs a for-each $(1\pm\eps)$-cut sketch of $G$
if there exists a recovering function $f$ such that, for each $S \subseteq V$, with probability at least $2/3$,

\medskip
$
(1-\eps) \cdot w(S, V \setminus S) \le f(S, g(G)) \le (1+\eps) \cdot w(S, V \setminus S).
$
\end{definition}

\section{For-All Sparsification: $\tilde O(n\cdot \beta/\epsilon^2)$ Upper Bound}\label{sec:for-all-ub}
In this section, we extend the seminal work of Benc\'ur and Karger~\cite{benczur2015randomized} to directed graphs using cut balance. For undirected graphs, they showed that sampling every edge inversely proportional to a quantity known as its strength (see Definition~\ref{def:strength}) preserves all cuts with high probability. We show that, by boosting this sampling probability by a factor of $\beta$, this procedure can preserve the value of $\beta$-balanced cuts in any directed graph.


We then show that this sampling theorem can be applied in a black-box manner to recover the analysis of a celebrated maximum flow algorithm for undirected graphs given by Karger and Levine~\cite{karger2002random}. At each step of the algorithm, they sample edges from the residual network (which is directed) of an undirected graph. Using a customized version of the sparsification result from Bencz\'ur and Karger~\cite{benczur2015randomized}, they show that with high probability, the sample contains an augmenting path. In contrast, our sampling procedure can be applied directly to the residual network, which simplifies the analysis of the algorithm.

The following theorem is our main result of this section:


\begin{theorem} \label{thm:beta-cuts-preserved}
Let $G = (V, E)$ be a directed graph where each edge $e$ has weight $u_e \geq 0$, and let $\epsilon, \beta$ be parameters. There is an $\tilde O(m)$-time algorithm that returns a weighted subgraph $H$ that satisfies the following with high probability: for every $\alpha$-balanced cut $U$,
\[
\left(1 - \epsilon\sqrt{\frac{\alpha +1}{\beta+1}}\right)\cdot \delta_G(U) \leq \delta_H(U) \leq \left(1 + \epsilon\sqrt{\frac{\alpha +1}{\beta+1}}\right)\cdot \delta_G(U).
\]
where $\delta_G(U)$ and $\delta_H(U)$ denote the cut value of $U$ in $G$ and $H$, respectively. Furthermore, $H$ contains $O(\beta n\log n/\epsilon^2)$ edges in expectation.
\end{theorem}

Note that for the special case where the graph $G$ is $\beta$-balanced, Theorem~\ref{thm:forall-strong} is implied by Theorem~\ref{thm:beta-cuts-preserved}: all cut values are preserved. This is the main result of Ikeda and Tanigawa~\cite{ikeda2018cut}.

\begin{corollary}[Ikeda and Tanigawa~\cite{ikeda2018cut}] \label{cor:cuts-preserved}
Consider the same setting as Theorem~\ref{thm:beta-cuts-preserved}. If $G$ is $\beta$-balanced, then with high probability, $H$ approximates every cut of $G$ up to a $(1 \pm \epsilon$) factor.
\end{corollary}


Before proving Theorem~\ref{thm:beta-cuts-preserved}, we give an application of digraph sparsification to the maximum flow problem. In particular, we prove the correctness of the $\tilde O(m+nv)$-time maximum flow algorithm given by Karger and Levine~\cite{karger2002random}, where $v$ is the value of the maximum flow. This algorithm (Algorithm~\ref{alg:max-flow}) is an adaptation of the classic augmenting paths algorithm of Ford and Fulkerson~\cite{ford1956maximal}, but with the following crucial observation. Let $f$ denote the current flow value in any iteration, and $\gamma = (v-f)/v$ denote the fraction of remaining flow in the residual network. Karger and Levine~\cite{karger2002random} show that, by boosting the undirected sampling procedure of Bencz\'ur and Karger~\cite{benczur2015randomized} by a factor of $1/\gamma$ and applying it to the residual network, the resulting sample contains an augmenting path with high probability. This saves on running time since the search for an augmenting path can then be performed on the sampled graph instead of the entire residual network.

\begin{algorithm}[ht]
\caption{A $\tilde O(m+nv)$-time algorithm for maximum flow (Karger and Levine~\cite{karger2002random}}
\label{alg:max-flow}
\SetKwInOut{Input}{Input}
\SetKwInOut{Output}{Output}
\Input{An undirected graph $G = (V, E)$ on $n$ vertices and $m$ edges with edge capacities $u_e$, source vertex $s$, sink vertex $t$.}
\Output{A maximum flow in $G$.}
Use Lemma~\ref{lem:strengths} to compute an estimated edge strength $\tilde{k}_e$ for every edge $e$. \\
Set $\alpha = 1$. \\
\While{$\alpha n < m$}{
Sample $\alpha n$ edges from the residual network (with replacement, but ignoring duplicates) according to weights $u_e / \tilde{k}_e$. \\
Find an augmenting path among the sampled edges.
\If{no path is found}{
Double the value of $\alpha$.
}
Find and augment paths in the residual network until none remain.
}
\Return{the resulting flow.}
\end{algorithm}

In contrast, we show that we can directly apply digraph sparsification to the residual network, with $\beta = 2/\gamma$ to obtain a short proof of the Karger-Levine theorem:

\begin{theorem}[Karger and Levine~\cite{karger2002random}]
Suppose we apply the algorithm in Theorem~\ref{thm:beta-cuts-preserved} to the residual network in a maximum flow computation, with $\epsilon = 0.1$ and $\beta = 2/\gamma$, where $\gamma = (v-f)/v$ is the fraction of flow remaining in the residual network. Then with high probability, there is an augmenting path in the sample.
\end{theorem}

\begin{proof}
We claim that every $s$-$t$ cut $S$ in the residual graph is $\beta$-balanced, where $\beta = 2/\gamma$. Suppose $S$ initially contains capacity $c \geq v$, and currently, $x$ units of flow are entering $S$. Since the flow value is $(1-\gamma)v$, the amount of flow leaving $S$ is $x + (1 - \gamma)v$. At the same time, the $x$ units of flow entering $S$ create a residual capacity of $x$ leaving $S$. Thus, the total residual capacity leaving $S$ is $c - x - (1-\gamma)v + x = c - (1-\gamma)v$. We can similarly show that the residual capacity entering $S$ is $c + x + (1-\gamma)v - x = c + (1-\gamma)v$. Thus, in the residual graph, the balance of $S$ is at most
$
\frac{c+(1-\gamma)v}{c-(1-\gamma)v} \leq \frac{2 - \gamma}{\gamma} \leq \frac{2}{\gamma}.
$

Now by setting $\epsilon = 0.1$ and $\beta = 2/\gamma$, Theorem~\ref{thm:beta-cuts-preserved} implies that the sparsifier $H$ preserves all $(2/\gamma)$-balanced cuts up to a $(1 \pm 0.1)$ factor with high probability. Since every $s$-$t$ cut is $(2/\gamma)$-balanced, this implies that there exists an augmenting path in $H$, as desired.
\end{proof}

In the rest of this section, we prove Theorem~\ref{thm:beta-cuts-preserved}. Before we give our algorithm, we state the definitions and results that we need from previous work.

\begin{definition}[Strength and strong components] \label{def:strength}
The \emph{strength} of an edge $e$, denoted by $k_e$, is the largest $k$ such that there exists a $k$-edge-connected vertex-induced subgraph of $G$ containing $e$. A \emph{$k$-strong component} is the subgraph induced by edges with strength at least $k$.
\end{definition}

\begin{lemma}[Bencz\'ur and Karger~\cite{benczur2015randomized}] \label{lem:components}
The strong components of an undirected graph form a laminar family, and a graph on $n$ vertices has at most $n-1$ nontrivial strong components.
\end{lemma}

\begin{lemma}[Bencz\'ur and Karger~\cite{benczur2015randomized}] \label{lem:strengths}
In any graph with edge weights $u_e$ and strengths $k_e$, we have $\sum_e u_e/k_e \leq n-1$. Furthermore, there exists an $O(m\log^3 n)$-time algorithm that returns, for every edge $e$, an estimate $\tilde{k}_e$ of $k_e$ satisfying $\tilde{k}_e \leq k_e$ and $\sum_e u_e / \tilde{k}_e = O(n)$. 
\end{lemma}


We now describe our algorithm (Algorithm~\ref{alg:sample}).
The input is a directed graph where each edge $e$ has weight $u_e$. We first compute approximate edge strengths $\tilde{k}_e$ as given in Lemma~\ref{lem:strengths}. Then we sample each edge $e$ proportional to $(\beta+1)u_e/\tilde{k}_e$, where $\beta\geq 1$ is a chosen parameter.
We choose the weight of the sampled edges so that we get an unbiased estimator.

\begin{algorithm}[!ht]
\DontPrintSemicolon
\caption{For-all sparsification for directed graphs}
\label{alg:sample}
\SetKwInOut{Input}{Input}
\SetKwInOut{Output}{Output}
\Input{An $n$-vertex directed graph $G = (V, E, u)$ with edge weights $u_e$, $0 < \eps < 1$, $\beta \geq 1$, and a constant $d > 2$.}
\Output{A subgraph $H$ that satisfies Theorem~\ref{thm:beta-cuts-preserved}.}
Use Lemma~\ref{lem:strengths} to compute an estimated edge strength $\tilde{k}_e \le k_e$ for every edge $e \in E$. \\
Let $\rho = 3d(\beta+1)\log n / \epsilon^2$. \\
\For{each edge $e\in E$}{
Sample $e$ with probability $p_e = \rho\cdot u_e/\tilde{k}_e$. \\
\lIf{$e$ is sampled}{add $e$ to $H$ with weight $w_e = \tilde{k}_e / \rho$.
}
}
\Return{$H$.}
\end{algorithm}

Now we analyze the output $H$ of Algorithm~\ref{alg:sample}. Without loss of generality, we assume that the algorithm uses the actual edge strengths $k_e$ rather than the estimates $\tilde{k}_e$. This is because $\tilde{k}_e \leq k_e$ and it does not hurt to oversample in importance sampling.

For each strong component $G_i$ of $G$ (see Definition~\ref{def:strength}), let $H_i$ denote the corresponding component in $H$. Because the way we choose sampling probabilities and edge weights in $H$, we have $\Ex{H_i} = G_i$. Let $\alpha_i=(k_i-k_{p(i)})/\rho$ where $p(i)$ is $G_i$'s parent in the laminar family formed by strong components (see Lemma~\ref{lem:components}). As shown by Bencz\'ur and Karger~\cite{benczur2015randomized}, this results in a decomposition of $G$ into its strong components, that is, $G = \sum_i\alpha_i G_i$.

For a component $G_i$ and a cut $U$, let $\delta_{G_i}(U)$ be the total capacity of edges leaving $U$ in $G_i$, and let $\dun{G_i}{U}$ be the corresponding value for the undirected version of $G_i$. The following lemma shows that for every strong component $G_i$, with high probability, $\delta_{G_i}(U)$ is preserved in $H$ up to a relative error for every cut $U$.

\begin{lemma} \label{lem:preserve-Gi}
Let $H$ be the output of Algorithm~\ref{alg:sample}.
For each strong component $G_i$ (defined in Definition~\ref{def:strength}), the following holds with probability at least $1 - O(n^{-d+2})$: for any cut $U$, we have $\abs{\delta_{H_i}(U)-\delta_{G_i}(U)} \leq \zeta(U) \cdot \delta_{G_i}(U)$ where $\zeta(U) = \epsilon\sqrt{\dun{G_i}{U}/(\delta_{G_i}(U)(\beta+1))}$.
\end{lemma}

\begin{proof}
For any cut $U$, let $\delta_p(U)$ denote the sum of $u_e/k_e$ over the (undirected) edges crossing $U_j$ in $G_i$. Order the $r$ cuts intersecting $G_i$ such that $1 = \delta_p(U_1) \leq \cdots \leq \delta_p(U_r)$, and let
\[
q_j = \Pr\left(\abs{\delta_{H_i}(U_j)-\delta_{G_i}(U_j)} > \zeta(U) \cdot \Ex{\delta_{H_i}(U_j)}\right).
\]
By a Chernoff bound, we have
\begin{equation} \label{eq:chern-qj}
q_j \leq 2\exp\left(-\frac{(\zeta(U))^2 \cdot \delta_{G_i}(U_j)}{3}\right) = 2\exp\left(-\frac{\epsilon^2 \cdot \dun{G_i}{U_j}}{3(\beta+1)}\right),
\end{equation}
where the equality follows substituting the definition of $\zeta(U)$. Let $E(i,j)$ denote the set of edges crossing $U_j$ in $G_i$. Since each edge $e$ in $G_i$ has weight $p_e$, we have
\[
\dun{G_i}{U_j} = \sum_{e\in E(i,j)} p_e = \sum_{e\in E(i,j)} \frac{3d(\beta+1)u_e}{\epsilon^2 k_e}\cdot \log n,
\]
Substituting this into Eq.~\eqref{eq:chern-qj} shows
\[
q_j \leq 2\exp\left(-d\sum_{e\in E(i,j)} \frac{u_e}{k_e} \cdot \log n \right) = 2n^{-d \cdot \delta_p(U_j)}.
\]
Since $\delta_p(U_j) \geq 1$, we have $q_j \leq 2n^{-d}$, so
\begin{equation} \label{eq:first-n2}
\sum_{j\leq n^2} q_j \leq n^2 \cdot 2n^{-d} = 2n^{-d+2} = O(n^{-d+2}).
\end{equation}
For $j \geq n^2$, we express $j$ as $j = n^{2\lambda}$. The number of $\lambda$-minimum cuts is at most $j$ (see, e.g., Karger and Stein~\cite{karger1996new}) and $\delta_p(U_1) = 1$, so $\delta_p(U_j) \geq \lambda$, so $\delta_p(U_j) \geq \frac{\log j}{2\log n}$. This implies, for $j \geq n^2$, $q_j \leq 2n^{-(d\log j) / (2 \log n)} = 2j^{-d/2}$.
Combining this with Eq.~\eqref{eq:first-n2}, we can conclude
\[
\sum_{j\ge 1}q_j \leq O(n^{-d+2}) +2\int_{n^2}^\infty j^{-d/2} dj = O(n^{-d+2}).\qedhere
\]
\end{proof}

Now we are ready to prove Theorem~\ref{thm:beta-cuts-preserved}.

\begin{proof}[Proof of Theorem~\ref{thm:beta-cuts-preserved}]
By Lemma~\ref{lem:strengths}, the expected number of edges in $H$ is $\sum_e p_e = O(\beta n\log n/\epsilon^2)$, as claimed. Now consider an $\alpha$-balanced cut $U$. We have
\[
\abs{\delta_H(U)-\delta_G(U)} = \Bigl|\sum_i\alpha_i\delta_{H_i}(U)-\alpha_i\delta_{G_i}(U)\Bigr| \leq  \sum_i\alpha_i\abs{\delta_{H_i}(U)-\delta_{G_i}(U)}.
\]
Taking a union bound over all strong components, we know that with high probability, Lemma~\ref{lem:preserve-Gi} holds for every strong component. Thus, the quantity above is at most
\begin{align*}
\sum_i \alpha_i\zeta(U) \cdot \delta_{G_i}(U) &= \sum_i \alpha_i\epsilon\sqrt{\dun{G_i}{U}\delta_{G_i}(U)/(\beta+1)} \\
&\leq \frac{\epsilon}{\sqrt{\beta+1}}\sqrt{\sum_i\alpha_i\dun{G_i}{U}\sum_i\alpha_i\delta_{G_i}(U)} \tag{Cauchy-Schwarz} \\
&= \frac{\epsilon}{\sqrt{\beta+1}}\sqrt{\dun{G}{U}\delta_G(U)}. \tag{$\sum_i\alpha_i\delta_{G_i}(U) = \delta_G(U)$}
\end{align*}
We conclude the proof by noting that $\dun{G}{U} \leq (\alpha+1)\cdot \delta_G(U)$, since $U$ is $\alpha$-balanced.
\end{proof}

\section{For-All Sparsification: $\Omega(n\cdot \beta/\epsilon)$ Lower Bound}
Our goal in this section to prove a lower bound whose dependence on $\beta$ matches the linear upper bound given by Ikeda and Tanigawa~\cite{ikeda2018cut} on the size of for-all cut sketches: 

\begin{theorem}
\label{thm:allcuts-lb}
Fix $\beta \ge 1$ and $0 < \eps < 1$ where $\beta/\eps \le n/2$.
Any $(1\pm\eps)$ for-all cut sketching algorithm for $n$-node $\beta$-balanced graphs must output $\Omega(n\cdot \beta / \eps)$ bits in the worst case.
\end{theorem}

We first prove a special case of our lower bound for $\beta = \Theta(n)$ and $\eps = \Theta(1)$ (Lemma~\ref{lem:allcuts-lb-simple}).
The proof for this special case contains the main ideas of our lower-bound construction for general values of $\beta$ and $\eps$.

\begin{lemma}
\label{lem:allcuts-lb-simple}
Let $\beta = 8n$ and let $\eps$ be a sufficiently small universal constant.
Any $(1\pm\eps)$ for-all cut sketching algorithm for $n$-node $\beta$-balanced graphs must output $\Omega(\beta n)$ bits in the worst case.
\end{lemma}

We give an overview of how we prove Lemma~\ref{lem:allcuts-lb-simple}. Without loss of generality, we can focus on deterministic cut sketching algorithms, because running time is not a concern in Lemma~\ref{lem:allcuts-lb-simple}, any randomized sketching algorithm can be derandomized by enumerating all possible coin flips.

We will choose a set of graphs $\calG$ such that the following conditions hold:
\begin{itemize}
\item Every graph in $\calG$ is $\beta$-balanced.
\item The size of $\calG$ is large (Lemma~\ref{lem:gnbeta-size}).
\item There exists a $\ell$ with $|\calG| / \ell = 2^{\Omega(\beta n)}$ such that, for every graph $G \in \calG$, there are at most $\ell$ graphs in $\calG$ that can share a $(1\pm\eps)$-cut sketch with $G$ (Lemma~\ref{lem:encode-H}).
\end{itemize}
This way, each cut sketch works for at most $\ell$ graphs in $\calG$, so any algorithm must produce at least $|\calG| / \ell = 2^{\Omega(\beta n)}$ different cut sketches for all graphs in $\calG$, which implies that the algorithm must output at least $\log_2 (|\calG|/\ell) = \Omega(\beta n)$ bits.

Formally, consider the set of graphs $\calG_{2n}$ with $2n$ vertices defined as follows: every graph $G \in \calG_{2n}$ is an unweighted bipartite graph with bipartitions $L, R$ satisfying $|L| = |R| = n$. Fix a perfect matching from $L$ to $R$. The set $\calG_{2n}$ is defined to contain all graphs $G$ such that the edges from $L$ to $R$ is exactly this perfect matching (and the set of edges from $R$ to $L$ are arbitrary).
Let $\calG_{2n,\beta} \subseteq \calG_{2n}$ be the subset of graphs in $\calG_{2n}$ that are $\beta$-balanced.


As described above, Lemma~\ref{lem:gnbeta-size} gives a lower bound on the size of $\calG_{2n,\beta}$.

\begin{lemma}
\label{lem:gnbeta-size}
Let $n_0$ be a sufficiently large universal constant.
If $n \ge n_0$ and $\beta = 8 n$, then $|\calG_{2n,\beta}| \ge 2^{n^2/2}$.
\end{lemma}

The next lemma upper bounds the maximum number of graphs in $\calG$ that can share an $(1\pm\eps)$-cut sketch.
Notice that if $G$ and $H$ have the same $(1\pm\eps)$-cut sketch, then $H$ must be a $(1\pm3\eps)$-cut sparsifier of $G$.

\begin{lemma}
\label{lem:encode-H}
Let $\eps > 0$ be a sufficiently small universal constant.
For every $G \in \calG_{2n}$, the number of graphs in $\calG_{2n}$ that are $(1\pm3\eps)$-cut sparsifiers of $G$ is at most $2^{n^2/4}$.
\end{lemma}

We now prove Lemma~\ref{lem:allcuts-lb-simple} using Lemmas~\ref{lem:gnbeta-size}~and~\ref{lem:encode-H}. The proofs of Lemmas~\ref{lem:gnbeta-size}~and~\ref{lem:encode-H} are in Sections~\ref{sec:gnbeta-size} and~\ref{sec:encode-H}, respectively.

\begin{proof}[Proof of Lemma~\ref{lem:allcuts-lb-simple}]
We work with graphs with $2n$ vertices (rather than $n$ vertices) to make the presentation easier.
This is equivalent because we aim to prove a lower bound of $\Omega(\beta n)$.

Fix any $(1\pm\eps)$ for-all cut sketching algorithm.
Consider running this algorithm on all graphs in $\calG_{2n,\beta}$. Every graph in $\calG_{2n,\beta}$ is $\beta$-balanced, so the algorithm must map every $G \in \calG_{2n,\beta}$ to a bit string (i.e., cut sketch), and graphs that are not $(1\pm3\eps)$-cut sparsifiers of each other must be mapped to different strings.
By Lemma~\ref{lem:gnbeta-size}, there are at least $2^{n^2/2}$ graphs in $\calG_{2n,\beta}$, and by Lemma~\ref{lem:encode-H}, at most $2^{n^2/4}$ graphs can be mapped to the same bit string.
Therefore, the algorithm must output at least $\frac{2^{n^2/2}}{2^{n^2/4}} = 2^{n^2/4}$ distinct bit strings.
This implies that the algorithm must output at least $\frac{n^2}{4} = \frac{1}{32} \beta n$ bits in the worst case.
\end{proof}

\subsection{Proof of Lemma~\ref{lem:gnbeta-size}}
\label{sec:gnbeta-size}

In this section we prove Lemma~\ref{lem:gnbeta-size}. We first prove the following lemma (Lemma~\ref{lem:rand-balanced}),  which states that most of the graphs in $\calG_{2n}$ are balanced.
Lemma~\ref{lem:gnbeta-size} follows immediately from Lemma~\ref{lem:rand-balanced}, because $|\calG_{2n}| = 2^{n^2}$ and we have $|\calG_{2n,\beta}| \ge (1-\frac{1}{n})|\calG_{2n}| \ge 2^{n^2/2}$.

\begin{lemma}
\label{lem:rand-balanced}
Fix $n \ge n_0$ where $n_0$ is a sufficiently large universal constant.
Consider a graph $G$ drawn uniformly from $\calG_{2n}$.
With probability at least $1 - \frac{1}{n}$, $G$ is $\beta$-balanced for $\beta = 8 n$.
\end{lemma}

To prove Lemma~\ref{lem:rand-balanced}, we will establish a set of deterministic conditions (Lemmas~\ref{lem:rand-deg}~and~\ref{lem:rand-largesets}) and show that these conditions hold with high probability; together, they imply that $G$ is balanced.

Intuitively, these conditions correspond to two special types of cuts.
Lemma~\ref{lem:rand-deg} states that the in- and out-degrees of every vertex in $G$ behave as expected, which implies all singleton cuts are $\beta$-balanced.
Lemma~\ref{lem:rand-largesets} states that for large sets $A \subseteq L$ and $B \subseteq R$, the number of edges from $B$ to $A$ is as expected, which implies any cut $S$ is balanced if both $S \cap L$ and $S \cap R$ are not too large or too small.
As we will see in the proof of Lemma~\ref{lem:rand-balanced}, it turns out these conditions not only imply the balance of the above cuts, they are sufficient to imply the balance of all cuts in $G$.

Formally, the first lemma shows that with high probability, the in-degree of every vertex $u \in L$ and the out-degree of every vertex $v \in R$ are concentrated around their expectations.

\begin{lemma}
\label{lem:rand-deg}
Fix $n \ge n_0$ where $n_0$ is a sufficiently large universal constant.
For $G$ drawn uniformly from $\calG_{2n}$, with probability at least $1 - \frac{1}{n^2}$, we have
\begin{enumerate}[(i)]
\item $\frac{3}{8} n \le |E(R, u)| \le \frac{5}{8}n$ for every $u \in L$, and
\item $\frac{3}{8} n \le |E(v, L)| \le \frac{5}{8}n$ for every $v \in R$.
\end{enumerate}
\end{lemma}
\begin{proof}
We prove part $(i)$; the proof of $(ii)$ follows similarly.
Fix any $u \in L$.
For each $v$, the event $(v, u) \in E$ happens with probability $\frac{1}{2}$ independently.
Therefore, the expectation of $|E(R, u)|$ is $\frac{n}{2}$, and by a Chernoff bound, $\Pr\left[\left| |E(R, u)| - \frac{n}{2} \right| > \frac{1}{8} n\right] \le 2 \exp\bigl(-\frac{(1/8)^2 n/2}{3}\bigr) = \exp(-\Omega(n))$.
Part $(i)$ follows from $n \ge n_0$ and taking the union bound over all $u \in L$.
\end{proof}

The next lemma shows that with high probability, the number of edges from $B$ to $A$ is at least half of its expectation for all large subsets $A \subseteq L$ and $B \subseteq R$.

\begin{lemma}
\label{lem:rand-largesets}
Fix $n \ge n_0$ where $n_0$ is a sufficiently large universal constant.
For $G$ drawn uniformly from $\calG_{2n}$, with probability at least $1 - \frac{1}{n^2}$, we have
\[
|E(B, A)| \ge \frac{1}{4} |B| |A|
\]
for every $A \subseteq L, B \subseteq R$ satisfying $|A|, |B| \ge \frac{n}{8}$.
\end{lemma}
\begin{proof}
Let $a = |A|$ and $b = |B|$.
Fix $\frac{n}{8} \le a, b \le n$.
For a specific pair of sets $(A, B)$ of size $(a, b)$, the expectation of $|E(B,A)|$ is $\frac{a b}{2}$.
By a standard application of the Chernoff bound, the probability that the condition in the lemma fails for this pair of $(A, B)$ is at most $\exp(-\Omega(a b))$.
On the other hand, the total number of such $(A, B)$ pairs is at most $n^{a + b}$.

Taking the union bound over all possible sets $(A, B)$, the probability that any $|E(B, A)|$ deviates too much is at most $\sum_{a,b}n^{a+b} \exp\left(-\Omega(ab)\right)$.
When $a, b \ge \frac{n}{8}$ and $n \ge n_0$, the failure probability is at most $n^2 \cdot n^{2n} \cdot \exp(-\Omega(n^2)) = \exp(-\Omega(n^2)) \le \frac{1}{n^2}$.
\end{proof}

Assuming the high probability events in Lemmas~\ref{lem:rand-deg}~and~\ref{lem:rand-largesets} happen, we are now ready to prove Lemma~\ref{lem:rand-balanced}.

\begin{proof}[Proof of Lemma~\ref{lem:rand-balanced}]
Recall that $\beta = 8n$ and $G = (L \cup R, E)$ is an unweighted bipartite graph with a perfect matching from $L$ to $R$.
For this proof, we assume the edges from $R$ to $L$ satisfy the conditions stated in Lemmas~\ref{lem:rand-deg}~and~\ref{lem:rand-largesets}, which happens with probability at least $1 - \frac{1}{n}$.

Fix any cut $S \subseteq V$.
We will show that $S$ is $\beta$-balanced.
That is, the total weight of edges leaving and entering $S$ are within a factor of $\beta$ of each other.
Suppose $S = A \cup B$ where $A \subseteq L$ and $B \subseteq R$, and let $\overline A = L \setminus A$ and $\overline B = R \setminus B$. Observe that the set of edges leaving $S$ is $E(S, \overline S) = E(B, \overline A) \cup E(A, \overline B)$, and the set of edges entering $S$ is $E(\overline S, S) = E(\overline A, B) \cup E(\overline B, A)$.

If $A = \varnothing$, then $S$ is $\beta$-balanced because $E(\overline S, S) = |B|$ (from the perfect matching) while $\frac{3}{8} n |B| \leq |E(S, \overline S)| \leq \frac{5}{8}n |B|$ due to Lemma~\ref{lem:rand-deg} (applied to each vertex in $B$). A similar argument holds for the case of $B = \varnothing$, so for the rest of the proof, we assume $A \neq \varnothing$ and $B \neq \varnothing$.

To prove that $S$ is $8n$-balanced, it is sufficient to show that $S$ has at least $\frac{1}{8}n$ outgoing edges and at least $\frac{1}{8}n$ incoming edges. This is the number of edges in either direction is at most $n^2$, so their ratio is at most $8n$.

First we assume $|A| \le \frac{1}{8} n$ and show that both $|E(S, \overline{S})|$ and $|E(\overline{S}, S)|$ are at least $\frac{1}{8}n$.
\begin{figure}[h]
\centering
\begin{subfigure}[t]{0.3\textwidth}
    \centering
        \includegraphics[width=0.9\textwidth]{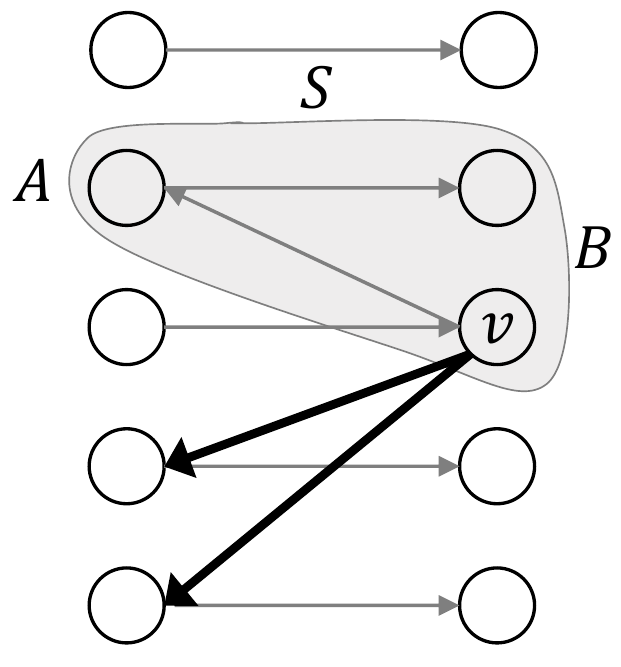}
        \caption{(Outgoing edges) Any $v \in B$ has at least $n/4$ edges leaving $S$ because outgoing degree of $v$ is large but $A$ is small.}
        \label{fig:outgoing}
\end{subfigure}
\hfill
\begin{subfigure}[t]{0.3\textwidth}
    \centering
        \includegraphics[width=0.9\textwidth]{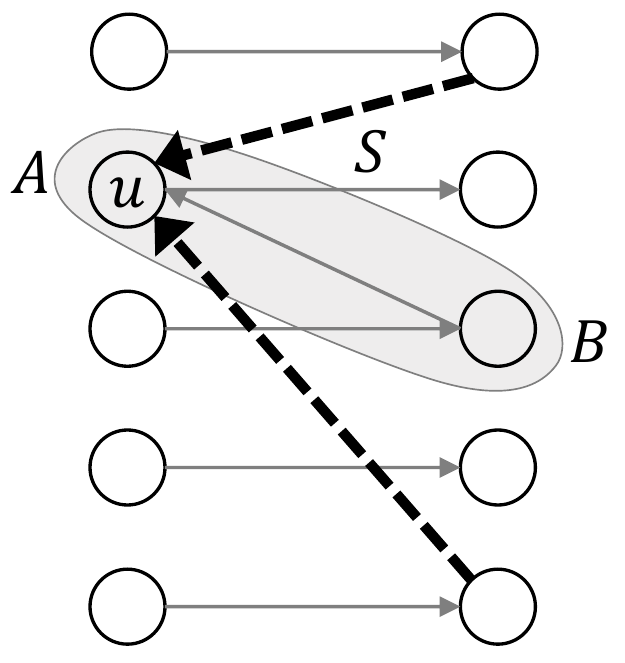}
        \caption{(Incoming edges, small $B$) Any $u\in A$ has at least $n/8$ edges entering $S$ because incoming degree of $u$ is large but $B$ is small.}
        \label{fig:incoming-small-B}
\end{subfigure}
\hfill
\begin{subfigure}[t]{0.3\textwidth}
    \centering
        \includegraphics[width=0.9\textwidth]{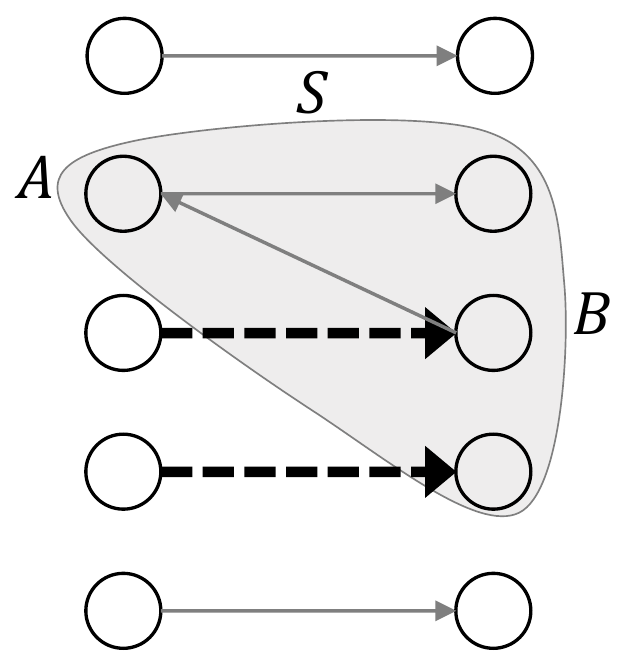}
        \caption{(Incoming edges, large $B$) There are at least $n/8$ matching edges entering $S$, because $B$ is much larger than $A$.}
        \label{fig:incoming-large-B}
\end{subfigure}
\caption{In these figures, we lower bound the number of outgoing/incoming edges of $S$ when $A$ is small. More specifically, we lower bound the number of outgoing/incoming edges by the number of bold colored edges.}
\end{figure}

\begin{itemize}
    \item (Outgoing edges) Consider any vertex $v \in B$. By Lemma~\ref{lem:rand-deg}, we have $|E(S, \overline S)| \ge |E(B, \overline A)| \ge |E(v, \overline A)| \ge |E(v, L)| - |A| \ge \frac{3}{8}n - \frac{1}{8}n = \frac{1}{4} n$ (see Fig.~\ref{fig:outgoing}).
    \item (Incoming edges, small $B$) When $|B| < \frac{1}{4}n$, we can pick any $u \in A$ and by Lemma~\ref{lem:rand-deg}, we have $|E(\overline S, S)| \ge |E(\overline B, A)| \ge |E(\overline B, u)| \ge |E(R, u)| - |B| > \frac{3}{8}n - \frac{1}{4}n = \frac{1}{8}n$ (see Fig.~\ref{fig:incoming-small-B}).
    \item (Incoming edges, large $B$) When $|B| \ge \frac{1}{4}n$, we have $E(\overline S, S) \ge E(\overline A, B) \ge |B| - |A| \ge \frac{1}{4}n - \frac{1}{8}n = \frac{1}{8}n$ (see Fig.~\ref{fig:incoming-large-B}).
\end{itemize}

The case analysis above shows that any cut $S$ with $|A| \le \frac{1}{8}n$ is $\beta$-balanced.
Consequently, any cut $S$ with $|A| \ge \frac{7}{8}n$ is also $\beta$-balanced, because any $S$ and $\overline S$ have the same balance factor.

Moreover, by symmetry, we can show that whenever $|B| \le \frac{1}{8}n$ or $|B| \ge \frac{7}{8}n$ the cut is $\beta$-balanced. Flipping the orientation of every edge does not affect the balance of any cut, but allows us to swap $A$ and $B$ in the above arguments.

Finally, we are left with the case that $\frac{1}{8} n \le |A|, |B| \le \frac{7}{8} n$.
In this case, Lemma~\ref{lem:rand-largesets} applies to both $E(B, \overline A)$ and $E(\overline B, A)$, so in either direction we have at least $\frac{1}{4} \cdot \frac{n}{8} \cdot \frac{n}{8} = \frac{1}{256} n^2 \ge \frac{1}{8}n$ edges, and therefore, the cut is $\beta$-balanced.
\end{proof}

\subsection{Proof of Lemma~\ref{lem:encode-H}} \label{sec:encode-H}
In this section, we prove Lemma~\ref{lem:encode-H}, which states that for any graph $G \in \calG_{2n,\beta}$, there are at most $2^{n^2/4}$ graphs in $\calG_{2n,\beta}$ that are $(1\pm3\eps)$-cut sparsifiers of $G$.

\begin{proof}[Proof of Lemma~\ref{lem:encode-H}]
Let $H$ be a $(1\pm3\eps)$-cut sparsifier of $G$.
We first show that $G$ and $H$ must share many edges in common.
Fix any vertex $v \in R$.
Let $N_G(v)$ and $N_H(v)$ denote the set of (out-)neighbors of $v$ in $G$ and $H$ respectively.
Let $S = S_v = \{v\} \cup N_G(v)$ and $T = T_v = \{v\} \cup N_H(v)$. We will prove that $|T \setminus S| \le 3\eps n$.

\begin{figure}[ht]
\centering
\begin{subfigure}[t]{0.42\textwidth}
    \centering
        \includegraphics[width=0.8\textwidth]{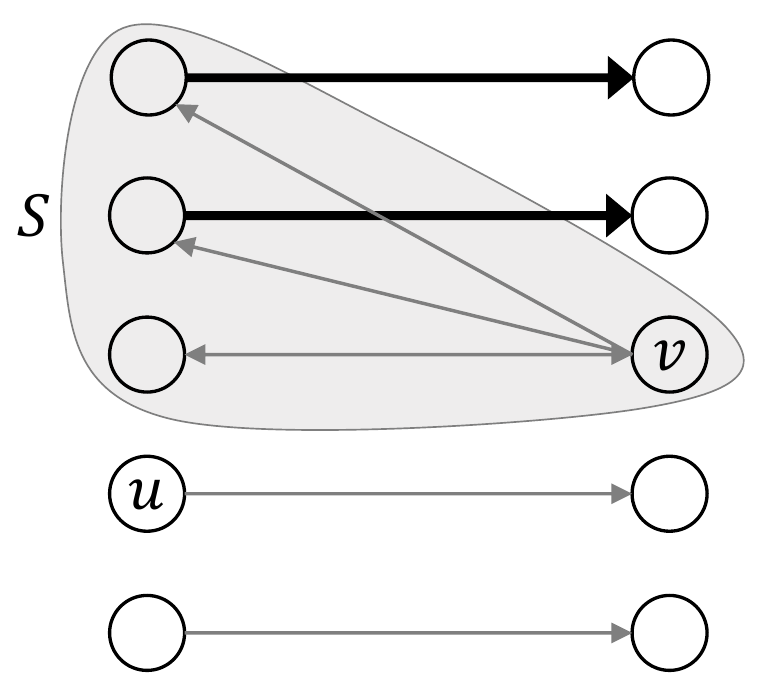}
        \caption{In graph $G$, the only edges leaving $S$ are the perfect matching edges.}
        \label{fig:encode-G}
\end{subfigure}
\hspace{0.06\textwidth}
\begin{subfigure}[t]{0.42\textwidth}
    \centering
        \includegraphics[width=0.8\textwidth]{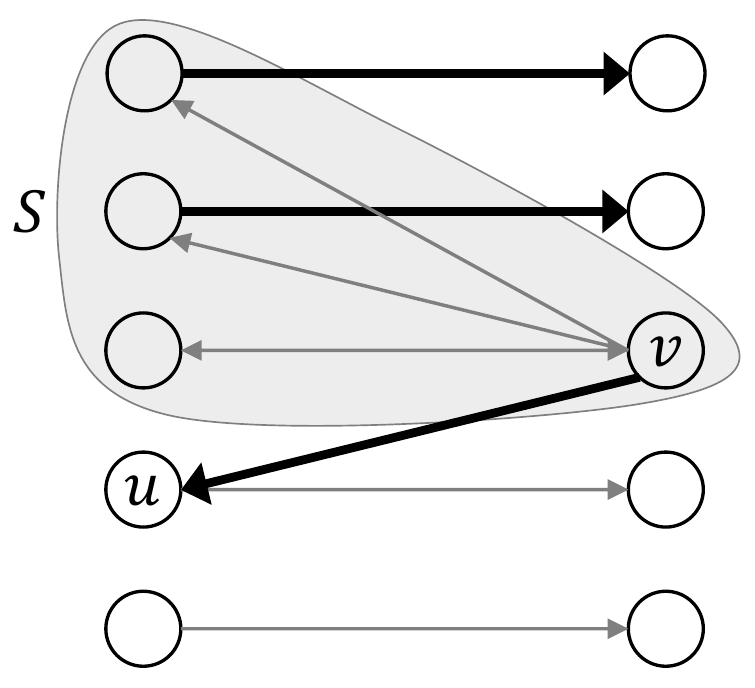}
        \caption{In graph $H$, for every vertex $u$ is in $T \setminus S$, $(v,u)$ contributes to $|E_H(S, \bar S)|$.}
        \label{fig:encode-H}
\end{subfigure}
\caption{The graph $G$ is given in Fig.~\ref{fig:encode-G} and $H$ is given in Fig.~\ref{fig:encode-H}.
$S$ is the union of $v$ and its neighbors in $G$ and $T$ is the union of $v$ and its neighbors in $H$.
For every $u \in T \setminus S$, the edge $(v, u)$ leaves $S$ and it exists in $H$ but not in $G$.
Since $H$ is a $(1 \pm 3\epsilon)$-cut sparsifier of $G$, $|T \setminus S| \leq 3 \epsilon n$.}
\label{fig:encode}
\end{figure}

Consider the number of outgoing edges from $S$ in both graphs.
In graph $G$, because $v$ has no edges leaving $S$, all outgoing edges from $S$ must be matching edges and hence $E_G(S, \overline S) \le n$.
In graph $H$, $E_H(S, \overline S)$ contains all edges in $E_G(S, \overline S)$, and in addition, one edge $(v, u)$ for each vertex $u \in T \setminus S$ (see Fig.~\ref{fig:encode}).
Because $H$ is a $(1\pm3\eps)$-cut sparsifier of $G$, we have 
\[
|T \setminus S| = E_H(S,\overline S) - E_G(S,\overline S) \le 3\eps E_G(S,\overline S) \le 3\eps n.
\]

We can similarly show that $|S \setminus T| \le 3\eps n$ by swapping $G$ and $H$ in this argument. So given a graph $G \in \calG_{2n}$, we can encode all graphs $H \in \calG_{2n}$ that are $(1\pm3\eps)$-cut sparsifier of $G$ as follows: for every $v \in R$, we encode $N_H(v)$ by writing down the vertices in $S_v \setminus T_v$ and the vertices in $T_v \setminus S_v$.
The number of possible choices for $S_v \setminus T_v$ is at most $\binom{n}{3\eps n} \le (\frac{e}{3\eps})^{3\eps n} \le 2^{n/8}$ when $\eps$ is a sufficiently small constant, and the same is true for $T_v \setminus S_v$.

There are $n$ vertices $v \in R$, and for each $v$, there are at most $2^{n/8}$ possible choices for either $S_v \setminus T_v$ and $T_v \setminus S_v$.
Therefore, we can upper bound the total number of possible $H$ by $\bigl((2^{n/8})^2\bigr)^n = 2^{n^2/4}$.
\end{proof}

\subsection{Proof of Theorem~\ref{thm:allcuts-lb}}

Now we prove our for-all cut sketch lower bound for general values of $\beta$ and $\eps$.

\begin{proof}[Proof of Theorem~\ref{thm:allcuts-lb}]
We will construct a set $\calG$ of $O(\beta)$-balanced graphs such that, for every graph $G \in \calG$, there are at most $\ell$ graphs in $\calG$ that can share a $(1 \pm c \cdot \eps)$ for-all cut sketch with $G$, where $c$ is a sufficiently small universal constant and $\ell$ satisfies $|\calG| / \ell = 2^{\Omega(n \beta / \eps)}$.

Let $\calG_{2k}$ denote the set of digraphs with $2k$ vertices defined as follows.
Every graph in $\calG_{2k}$ is a bipartite directed graph with bipartitions $L$, $R$ satisfying $|L| = |R| = k$.
There is a perfect matching from $L$ to $R$ with weight $\frac{1}{\eps}$ and every graph in $\calG_{2k}$ has the same perfect matching.
There are no other edges from $L$ to $R$.
The edges from $R$ to $L$ are arbitrary and they have unit weight.
Let $\calG_{2k,8\beta}$ be the subset of graphs in $\calG_{2k}$ that are $(8\beta)$-balanced.

We now describe the set $\calG$.
Let $k = \beta / \eps$.
Without loss of generality, we assume $k$ is an integer, $k \ge n_0$ where $n_0$ is the universal constant in Lemma~\ref{lem:gnbeta-size}, and $n$ is a multiple of $k$.
We partition the $n$ vertices into $t = n/k$ clusters of size $k$, denoted by $V_1, \ldots, V_t$.
For every $i \in \{1, \ldots, t-1\}$, we put a graph $G_i \in \calG_{2k,8\beta}$ between $V_i$ and $V_{i+1}$.
The set $\calG$ contains all possible graphs that can be constructed in this way.
Notice that we immediately have
\[
|\calG| = |\calG_{2k,8\beta}|^{t-1}.
\]

We first show that every graph $G \in \calG$ is $(8\beta)$-balanced.
Fix any $G \in \calG$.
Because each $G_i$ is a strongly connected graph between $V_i$ and $V_{i+1}$, the entire graph $G$ is strongly connected.
By the definition of $\beta$-balanced graphs, when we take the (edge) union of two $\beta$-balanced subgraphs, as long as the resulting graph is strongly connected, it is guaranteed to be $\beta$-balanced as well.

Next we show that for any graph $G \in \calG$, the number of graphs that can share a $(1 + c \cdot \eps)$ for-all cut sketch with $G$ is at most $(2^{k^2/4})^{t-1}$.
The argument is similar to that in the proof of Lemma~\ref{lem:encode-H}.
The main difference is that every matching edge now has weight $(1/\eps)$ and the cut queries are answered with precision $(1 \pm c \cdot \eps)$.

Let $H$ be a $(1 \pm c \cdot \eps)$-cut sparsifier of $G$.
We will show that given $G$, $H$ can be encoded using a small number of bits.
Let $(H_i)_{i=1}^{t-1}$ denote the corresponding bipartite subgraphs of $H$.
Fix any $i \in \{1, \ldots, t-1\}$ and $v \in V_{i+1}$.
Let $N_{G_i}(v) \subseteq V_i$ denote the set of out-neighbors of $v$ in $G_i$, and define $N_{H_i}(v)$ similarly.
Let
\[
S = \bigcup_{j=1}^{i-1} V_i \cup N_{G_i}(v) \cup \{v\}.
\]
Observe that in $G$, the edges leaving $S$ are
\begin{itemize}
    \item $(k - |N_{G_i}(v)|)$ matching edges from $V_{i-1}$ to $V_{i}$ (if $i > 1$).
    \item $|N_{G_i}(v)|$ or $|N_{G_i}(v)| - 1$ matching edges from $V_i$ to $V_{i+1}$.
    \item One edge from $v$ to $V_{i+2}$ (if $i < t-1$).
\end{itemize}

The total weight of these edges is at most $(k+1)\cdot\frac{1}{\eps} = O(k/\eps)$.
In $H$, the set of edges leaving $S$ includes the above edges, and in addition, one unit-weight edge $(v, u)$ for each vertex $u \in N_{H_i}(v) \setminus N_{G_i}(v)$.
Because $H$ is a $(1 \pm c \cdot \eps)$-sparsifier of $G$, the total weight of edges leaving $S$ in $G$ and $H$ can be off by at most $c \cdot \eps \cdot O(k/\eps) = O(ck)$.
In other words, $| N_{H_i}(v) \setminus N_{G_i}(v) | = O(k/\eps)$.
Consequently, the number of possible choices for the set $N_{H_i}(v) \setminus N_{G_i}(v)$ is at most $\binom{k}{O(c k)} \le 2^{k/8}$ when $c$ is sufficiently small (and the same bound holds for the choice of $N_{G_i}(v) \setminus N_{H_i}(v)$).
The choice for every $1 \le i < t$ and each $v \in V_{i+1}$ is independent, so the total number of possible $H$ is at most $\left((2^{k/4})^{k}\right)^{t-1}$.

The lower bound on the size of $\calG_{2k,8\beta}$ in Lemma~\ref{lem:gnbeta-size} continues to hold in this construction.
This is because there are in expectation $k^2 / 2 = \beta^2/ (2\eps^2)$ unit-weight edges in $G_i$ from $V_{i+1}$ to $V_i$, and $k$ matching edges from $V_i$ to $V_{i+1}$ whose total weight is $k \cdot (1/\eps) = \beta / \eps^2$.
Therefore, the balance of the cut $(V_i, V_{i+1})$ in $G_i$ is $\beta/2$ in expectation.
As in the proof of Lemma~\ref{lem:rand-balanced}, we can show that, if we include each edge from $V_{i+1}$ to $V_i$ independently with probability $1/2$, then with high probability (in $k$), the balance of $G_i$ is within a constant factor of its expectation (i.e., $G_i$ is $8 \beta$-balanced).
Consequently, there are at least $(1-\frac{1}{k}) \cdot 2^{k^2} \ge 2^{k^2/2}$ choices for each $G_i$, and the total number of graphs in $\calG$ is at least $(2^{k^2/2})^{t-1}$.

Putting everything together, a $(1\pm c\cdot \eps)$ for-all cut sketch has to output at least
\[
\log\left(\frac{(2^{k^2/2})^{t-1}}{(2^{k^2/4})^{t-1}}\right) =
 \Omega(k^2 t) = \Omega\left(\frac{\beta^2}{\eps^2} \cdot \frac{n}{\beta/\eps}\right) = \Omega(n \beta / \eps)
\]
bits for all graphs in $\calG$.
\end{proof}

\begin{remark}
{\em
Our analysis refutes a conjecture of Ikeda and Tanigawa~\cite{ikeda2018cut} that in a directed graph, importance sampling the edges with probability proportional to $1/\gamma_e$, where $\gamma_e$ is the {\em directed} edge connectivity of $e$, produces a directed cut sparsifier.

Consider an unweighted bipartite graph $G$ with bipartitions $|L|=|R|=n/2$.
There is a perfect matching from $L$ to $R$ and a complete graph from $R$ to $L$.
One can verify that the $n/2$ matching edges from $L$ to $R$ each have directed connectivity $1$, and the $n^2/4$ edges from $R$ to $L$ each have directed connectivity $\frac{n}{2}$ or $(\frac{n}{2} - 1)$.
Therefore, if we perform importance sampling with probability proportional to $\log^c(n)/\gamma_e$ for some constant $c$, the expected number of edges in the sampled graph is nearly-linear in $\sum_e 1/\gamma_e = \frac{n}{2} \cdot 1 + \frac{n^2}{4} \cdot \Theta(\frac{1}{n}) = O(n)$.

Let $H$ be the sampled graph.
Note that in $H$, every edge of the original perfect matching has weight 1 because those edges have directed connectivity 1 in $G$.
Take any vertex $v \in R$ such that the out-degree of $v$ in $H$ is $o(n)$.
Such a vertex is guaranteed to exist because $|E(H)| = o(n^2)$.
Let $S = \{v\} \cup N_H(v)$ and consider the total weight leaving $S$. Notice that $|w_G(S, \overline S) - \frac{n}{2}| \le 1$ and $\left| w_H(S, \overline S) - |S| \right| \le 1$.
This implies $w_G(S, \overline S) = \Omega(n)$ while $w_H(S, \overline S) = o(n)$, so $H$ cannot be a cut sparsifier of $G$.
}
\end{remark}


\section{For-Each Cut Sketch: $\tilde O(n\cdot \sqrt{\beta}/\epsilon)$ Upper Bound} \label{sec:for-each-ub}
\providecommand{\expect}[2]{\ensuremath{\ifthenelse{\equal{#1}{}}{\mathbb{E}}{\mathbb{E}_{#1}}\!\left[#2\right]}\xspace}

In this section, we give an upper bound on the size of cut sketches in the for-each setting. 

\begin{theorem} \label{thm:beta-12}
Let $G$ be an $n$-vertex $\beta$-balanced graph with edge weights in $[1, \poly(n)]$.
There exists a $(1\pm\eps)$ for-each cut sketch of size $O(n \beta^{1/2} \log^{3} n / \eps)$ bits that approximates the value $w(S, V \setminus S)$ of every directed cut $S \subseteq V$ with high probability.
\end{theorem}

We will prove a lower bound of $\Omega(n \cdot (\beta / \eps)^{1/2})$ bits in Section~\ref{sec:for-each-lb}.

\medskip
{\noindent \bf Overview of Our Approach.}
Our approach is inspired by the cut sketching algorithm for undirected graphs by Andoni et al.~\cite{andoni2016sketching}.
We first partition the edges into $\ell = O(\log n)$ disjoint sets $(E_i)_{i=1}^\ell$ based on their weights.
Edges in $G_i = (V, E_i, w)$ have roughly the same weight and we can essentially treat $G_i$ as an unweighted graph.
For each graph $G_i$, we ignore edge directions, and iteratively remove and store edges belonging to sparse cuts.

Note that $G_i$ may not be balanced.
Even when $G_i$ is balanced, the dense components of $G_i$ may not be balanced.
Despite this, we show that we can estimate the dense components' contribution to the cut value via random sampling.
This is because we can bound the variance within each component, and in the end, upper bound their sum (i.e., the overall variance) using the $\beta$-balance condition.

One of our main technical contributions is to derive a tighter upper bound on the variance of random sampling.
If we trace our analysis back to the undirected case, we remove some redundant terms in the analysis of~\cite{andoni2016sketching}.
This tighter variance bound is critical, because we cannot obtain the right space dependence on $\beta$ without it (see Appendix~\ref{sec:comp-andoni}).
In addition, our new analysis can be traced back to the undirected case, which will simplify the algorithm of~\cite{andoni2016sketching}.
We can obtain a for-each sketching algorithms for undirected graphs without downsampling or low-accuracy for-all sparsifiers, and the output is a graph.



\subsection{Sketching and Recovery Algorithms} \label{sec:beta-12}
Let $G = (V, E, w)$ be an $n$-node $\beta$-balanced directed graph with edge weights $w_e \in [1, \poly(n)]$.
It is worth noting that, when constructing the cut sketch, we do not know the cut query $S\subseteq V$.
The cut query $S$ is only given as input to the recovery algorithm.



\medskip
{\noindent \bf The Sketching Algorithm.}
We describe our overall cut sketching algorithm (Algorithm~\ref{alg:store}).
We partition the edges into $O(\log n)$ weight classes.
For each weight class, we iteratively store and remove all edges that belong to some $\lambda$-sparse cut (defined in Equation~\ref{eqn:lambda-sparse}).
When there are no $\lambda$-sparse cuts remain, we sample $\alpha$ incoming and outgoing edges at each vertex among the remaining edges. The values of $\lambda$ and $\alpha$ will be specified later in our analysis.

For a directed graph $G = (V, E, w)$, we say a cut $(S, \overline S)$ is \emph{$\lambda$-sparse} if the following holds:
\begin{align}
\label{eqn:lambda-sparse}
|E(S,\overline S)| + |E(\overline S,S)| & \le \lambda \cdot \min\bigl(| S |, | \overline S| \bigr).
\end{align}


\begin{algorithm}[!ht]
\caption{Compute a $(1 \pm O(\epsilon))$ for-each cut sketch}
\label{alg:store}
\SetKwInOut{Input}{Input}
\SetKwInOut{Output}{Output}
\Input{An $n$-vertex $\beta$-balanced graph $G = (V, E, w)$ with edge weights $w_e \in [1, \poly(n)]$, and $0 < \eps < 1$.}
\Output{A $(1\pm O(\eps))$ for-each cut sketch $\sk(G)$ of size $\tilde O(n \beta^{1/2} / \eps)$.}
Set $\alpha = \lambda = \beta^{1/2}/\eps$. \\
Partition the edges into $\ell = O(\log n)$ weight classes $E_1, \ldots, E_\ell$ where $E_i = \{e: w_e \in [2^{i-1}, 2^i)\}$.

Each weight class $E_i$ defines a (possibly unbalanced) graph $G_i = (V, E_i, w)$. \\
\For{$i = 1$ {\bf to} $\ell$}{
  \While{there exists a $\lambda$-sparse cut (defined in Equation~\eqref{eqn:lambda-sparse}) in $G_i$}{
    Remove all edges (in both directions) in this cut and store them in $\sk(G_i)$.
  }
  In $\sk(G_i)$, store the (dense) components $\{V_{ij}\}_j$ of $G_i$. \\
  For every $V_{ij}$ and every $u \in V_{ij}$, store the number of (remaining) incoming and outgoing edges at $u$ in $G_i$, i.e., $d^{\mathrm{in}}_{ij}(u) = |E_i(V_{ij}, u)|$ and $d^{\mathrm{out}}_{ij}(u) = |E_i(u, V_{ij})|$. \\
  At each vertex $u \in V$, sample with replacement $\alpha$ edges from the (remaining) outgoing edges $(u, v)$ and store them in $\sk(G_i)$. Do the same for incoming edges. \label{alg:step8}
}
\Return{$\sk(G) = \bigcup_i \sk(G_i)$.}
\end{algorithm}

\medskip
{\noindent \bf The Recovery Algorithm.} Algorithm~\ref{alg:recovery-one-sketch} is our recovery algorithm that queries the cut sketch $\sk(G)$ (i.e., the output of Algorithm~\ref{alg:store}).
We first establish some notation.
Recall that Algorithm~\ref{alg:store} decomposes $G$ into $(G_i)_{i=1}^\ell$ according to the edge weights. Let $V_{i} = (V_{ij})_j$ denote the set of dense components in $G_i$ after we iteratively remove the sparse cuts in $G_i$.


\begin{algorithm}[!ht]
\caption{Query the cut value $w(S, \overline S)$ from $\sk(G)$.}
\label{alg:recovery-one-sketch}
\SetKwInOut{Input}{Input}
\SetKwInOut{Output}{Output}
\setcounter{AlgoLine}{0}
\Input{A cut query $S \subseteq V$ and a cut sketch $\sk(G)$ (output of Algorithm~\ref{alg:store}).}
\For{each $\sk(G_i)$ in $\sk(G)$}{
    \For{each dense component $V_{ij}$ in $G_i$}{
        Let $S_{ij}$ denote the smaller set of $(V_{ij} \cap S)$ and $(V_{ij} \cap \overline S)$. \\
        \uIf{$S_{ij} = V_{ij} \cap S$}{
            Estimate the total weight of edges leaving $S_{ij}$:
            For every $u \in S_{ij}$, set
            \begin{equation} \label{eq:out-est-at-u}
            I_{S_{ij}}(u) = \frac{d^{\mathrm{out}}_{ij}(u)}{\alpha} \sum_{q=1}^\alpha \chi_{ij}(u,q)w_{ij}(u,q)
            \end{equation}
            where $d^{\mathrm{out}}_{ij}(u)$ is the out-degree of $u$ in $V_{ij}$, $\chi_{ij}(u,q)=1$ if the $q$-th sampled outgoing edge at $u$ crosses $S$ and $\chi_{ij}(u,q) = 0$ otherwise, and $w_{ij}(u,q)$ is the weight of the $q$-th sampled edge.
        }
        \Else{
            Estimate the total weight of edges entering $S_{ij} = V_{ij} \cap \overline S$ instead: \\
            For every $u \in S_{ij}$, set $I_{S_{ij}}(u)$ as in \eqref{eq:out-est-at-u}, using $d^{\mathrm{in}}_{ij}(u)$ instead of $d^{\mathrm{out}}_{ij}(u)$, and $\chi_{ij}(u,q)$ indicates if the $q$-th sampled incoming edge at $u$ crosses $S$.
        }
        The estimated contribution from $V_{ij}$ is $I_{V_{ij}} = \sum_{u \in S_{ij}} I_{S_{ij}}(u)$.
    }
    The estimated contribution from $G_i$ is $I_{G_i} = \sum_{V_{ij} \in G_i}I_{V_{ij}}$.
}
Compute $I_S = \sum_i I_{G_i}$, the estimate of the cut value from all dense-component edges. \\
Compute $J_S$, the total weight of $\lambda$-sparse cut edges that leaves $S$ in all $G_i$'s. \\
\Return{$I_S + J_S$.}
\end{algorithm}

Algorithm~\ref{alg:recovery-one-sketch} approximates $w(S,\overline S)$ by adding the total contribution of the sparse-cut edges and the dense-component edges. Let $J_S$ denote the total weight of sparse-cut edges that go from $S$ to $\overline S$ in all of the graphs $G_i$ (which we store deterministically). Let $I_S$ be the estimator for the total weight of dense-component edges leaving $S$ in all $G_i$ as defined in Algorithm~\ref{alg:recovery-one-sketch}.
Algorithm~\ref{alg:recovery-one-sketch} returns $I_S + J_S$ as the final answer.


\medskip
{\noindent \bf Correctness and Size Guarantees.}
We state the correctness of our recovery algorithm (Algorithm~\ref{alg:recovery-one-sketch}) in Lemma~\ref{lem:i-s-good} and the output size of our sketching algorithm (Algorithm~\ref{alg:store}) in Lemma~\ref{lem:sketch-size}.
Theorem~\ref{thm:beta-12} follows immediately from Lemmas~\ref{lem:i-s-good} and~\ref{lem:sketch-size}; we prove the latter before proving the former.

\begin{lemma}[Correctness of Algorithm~\ref{alg:recovery-one-sketch}] \label{lem:i-s-good}
Let $\sk(G)$ be the output of Algorithm~\ref{alg:store}.
Fix a cut query $S \subseteq V$. With probability at least $2/3$, the value $(I_S + J_S)$ returned by Algorithm~\ref{alg:recovery-one-sketch} on input $(S, \sk(G))$ satisfies
$
\abs{(I_S + J_S) - w(S, \overline S)} \le O(\eps) \cdot w(S, \overline S).
$
\end{lemma}

\begin{lemma}[Output Size of Algorithm~\ref{alg:store}]
\label{lem:sketch-size}
The output $\sk(G)$ of Algorithm~\ref{alg:store} has size $\tilde O(n (\lambda + \alpha)) = \tilde O(n \beta^{1/2} / \eps)$.
\end{lemma}

\begin{proof}
Without loss of generality, we assume $\eps = \Omega(1/n)$, otherwise we can store all edges exactly using $\tilde O(n/\eps)$ bits.
Algorithm~\ref{alg:store} produces $\ell = O(\log n)$ weight classes; each weight class defines a graph $G_i$.
In every $G_i$:
\begin{itemize}
    \item First we iteratively store and remove edges in $\lambda$-sparse cuts.
    We can upper bound the total number of edge removed using the following charging argument:
    When a $\lambda$-sparse cut is removed, we charge the cut size evenly to the vertices on the smaller side of the cut.
    Since the cut is $\lambda$-sparse, every vertex on the smaller side gets charged at most $\lambda$ edges.
    Each vertex can be charged at most $O(\log n)$ times because it can be in the smaller side $O(\log n)$ times.
    Therefore, $\sk(G)$ stores at most $O(\lambda n \log n)$ sparse edges, which takes $O(\lambda n \log^2 n)$ bits.
    \item On the remaining graph, the connected components are disjoint, so we can also store the partition of vertices into these dense components in $O(n \log n)$ bits.
    \item We can store the (remaining) in- and out-degree of every vertex in $O(n \log n)$ bits.
    \item We sample $O(\alpha)$ edges at each vertex in $V$, which requires $O(\alpha n \log n)$ bits.
\end{itemize}
Thus, for every $G_i$ we store $O(\lambda n \log^2 n + n \log n + \alpha n \log n) = O(n (\lambda + \alpha) \log^2 n)$ bits. Since $\alpha = \lambda = \beta^{1/2}/\eps$, the size of $\sk(G_i)$ is $O(n \beta^{1/2} \log^2 n / \eps)$.
The size of $\sk(G) = \bigcup_{i} \sk(G_i)$ is
\[
O(\log n) \cdot O(n \beta^{1/2} \log^2 n / \eps) = O(n \beta^{1/2} \log^3 n / \eps). \qedhere
\]
\end{proof}

Now we prove Lemma~\ref{lem:i-s-good} (the correctness of Algorithm~\ref{alg:recovery-one-sketch}). Note that it follows immediately from the following lemma and Chebyshev's inequality.

\begin{lemma} \label{lem:var-s}
The estimator returned in Algorithm~\ref{alg:recovery-one-sketch} is unbiased, i.e., $\Ex{I_S} + J_S = w(S, \overline S)$.
Moreover, the variance of $I_S$ is $\Var{I_S} \le O({\beta}/{\alpha \lambda}) w(S,\overline S)^2$.
\end{lemma}

\begin{proof}[Proof of Lemma~\ref{lem:i-s-good}]
Algorithm~\ref{alg:store} sets $\alpha = \lambda = \beta^{1/2} \eps^{-1}$, so Lemma~\ref{lem:var-s} implies $\Var{I_S} \le O(\eps^2) \cdot w(S,\overline S)^2$.
By Chebyshev's inequality, with probability at least $2/3$,
\[
\abs{(I_S + J_S) - w(S,\overline S)} \le O(\eps) \cdot w(S,\overline S). \qedhere
\]
\end{proof}

To prove Theorem~\ref{thm:beta-12}, all that remains is to prove Lemma~\ref{lem:var-s}.

\begin{proof}[Proof of Lemma~\ref{lem:var-s}]
Recall that our estimator is
\[
I_S = \sum_{G_i} \sum_{V_{ij}} \sum_{u \in S_{ij}} I_{S_{ij}}(u),
\]
where $i$ sums over the graphs $G_i$ defined according to the edge weights, $j$ sums over the dense components in each $G_i$ after all sparse cuts are removed, and $u$ sums over the vertices of $S_{ij}$. Without loss of generality, we can assume $|V_{ij} \cap S| \le |V_{ij} \cap \overline{S}|$ and hence $S_{ij} = V_{ij} \cap S$. Otherwise, Algorithm~\ref{alg:recovery-one-sketch} works with $V_{ij} \cap \overline{S}$ and queries for the incoming edges instead.
Under this assumption, we always work with outgoing edges:
\[
I_{S_{ij}}(u) = \frac{d^{\mathrm{out}}_{ij}(u)}{\alpha}\sum_{q=1}^\alpha \chi_{ij}(u,q) w_{ij}(u,q).
\]

Every edge $e \in E(S, \overline S)$ belongs to exactly one $G_i$, and in that $G_i$ it is either a sparse-cut edge, or a dense-component edge in exactly one $V_{ij}$.
Consequently, to prove $I_S + J_S$ is unbiased, it suffices to prove that $I_{S_{ij}}(u)$ is unbiased.
In the dense component $V_{ij}$ of $G_i$, the total contribution of edges leaving $u$ to $w(S, \overline S)$ is $\sum_{r=1}^{d^{\mathrm{out}}_{ij}(u)} \chi(u,r) \cdot w(u,r)$,
where $r$ indexes the edges leaving $u$, $w(u,r)$ is the weight of the $r$-th edge leaving $u$, and $\chi(u,r)$ indicates if this edge goes from $S$ to $\overline S$. Let $e(u,r)$ denote the $r$-th edge leaving $u$ and $e_{ij}(u,q)$ denote the $q$-th sampled edge leaving $u$ within $V_{ij}$.
Summing over the $\alpha$ sampled edges, we have
\begin{align*}
\Ex{I_{S_{ij}}(u)}
 &= \frac{d^{\mathrm{out}}_{ij}(u)}{\alpha} \cdot \sum_{q=1}^{\alpha} \Ex{\chi_{ij}(u,q) \cdot w_{ij}(u,q)} \\
 &= \frac{d^{\mathrm{out}}_{ij}(u)}{\alpha} \cdot \sum_{q=1}^\alpha \sum_{r=1}^{d^{\mathrm{out}}_{ij}(u)} \Pr[e_{ij}(u,q)=e(u,r)] \cdot \chi(u,r) \cdot w(u,r) \\
 &= \sum_{r=1}^{d^{\mathrm{out}}_{ij}(u)} \chi(u,r) \cdot w(u,r),
\end{align*}
where the last equality holds because each sample has the same variance, and the $q$-th sample is drawn uniformly among all outgoing edges at $u$, i.e., $\Pr[e_{ij}(u,q)=e(u,r)] = \frac{1}{d^{\mathrm{out}}_{ij}(u)}$. The expectation of $I_{S_{ij}}(u)$ is exactly the contribution of edges leaving $u$ to $w(S, \overline S)$, so $I_{S_{ij}}(u)$ is unbiased.

For the rest of the proof, we upper bound the variance of $I_S$.
We assume without loss of generality that $J_S = 0$, i.e., no sparse edges were ever stored and removed by Algorithm~\ref{alg:store}. This is because we are trying to prove the statement
\[
\Var{I_S} \le O\left(\frac{\beta}{\alpha \lambda}\right) w(S,\overline S)^2 = O\left(\frac{\beta}{\alpha \lambda}\right)\cdot (\Ex{I_S} + J_S)^2,
\]
so setting $J_S = 0$ only makes the right-hand side smaller and hence, the proof more difficult.

We introduce some notation: recall that $(V_{ij})_j$ is the set of dense components of $G_i$ and $S_{ij} = V_{ij} \cap S$.
We use $X_{ij} = |E_{i}(S_{ij}, \overline{S_{ij}})|$ to denote the \emph{number} of edges from $S_{ij}$ to $V_{ij} \cap \overline S$ in $G_i$, and $\overline{X_{ij}} = |E_{i}(\overline{S_{ij}}, S_{ij})|$ the number of edges in the reverse direction.
Let $X_i = \sum_{j} X_{ij}$ and $\overline{X_i} = \sum_j \overline{X_{ij}}$, so that $X_i$ is the total number of dense-component edges that go from $S$ to $\overline S$ in $G_i$.

Since there are no $\lambda$-sparse cuts (defined in \eqref{eqn:lambda-sparse}) at the end of Algorithm~\ref{alg:store}, we have $X_{ij} + \overline{X_{ij}} > \lambda \min(|S_{ij}|, |\overline{S_{ij}}|)$.
Since we assume $|S_{ij}| \leq |\overline{S_{ij}}|$, this condition implies the following for every dense component $V_{ij}$, $\lambda \abs{S_{ij}} \leq X_{ij} + \overline{X_{ij}}$.

Fix any $u \in S_{ij}$.
We first upper bound $\Var{I_{S_{ij}}(u)}$ and then work our way up the definition of $I_S$. Since $\chi_{ij}(u,q)$ is a Bernoulli random variable with mean $\abs{E_i(u, \overline{S_{ij}})}/d^{\mathrm{out}}_{ij}(u)$, its variance is
\[
\Var{\chi_{ij}(u,q)} = \frac{\abs{E_i(u, \overline{S_{ij}})}}{d^{\mathrm{out}}_{ij}(u)} \cdot \frac{\abs{E_i(u, S_{ij})}}{d^{\mathrm{out}}_{ij}(u)}.
\]
Now from the definition of $I_{S_{ij}}(u)$, we have
\begin{align*}
\Var{I_{S_{ij}}(u)} &= \frac{d^{\mathrm{out}}_{ij}(u)^2}{\alpha^2}\sum_{q=1}^\alpha \Var{\chi_{ij}(u,q)}w_{ij}(u,q)^2 \\
&\leq \frac{d^{\mathrm{out}}_{ij}(u)^2}{\alpha^2} \cdot \alpha \cdot \frac{\abs{E_i(u, \overline{S_{ij}})}}{d^{\mathrm{out}}_{ij}(u)} \cdot \frac{\abs{E_i(u, S_{ij})}}{d^{\mathrm{out}}_{ij}(u)}\cdot 2^{2i} \tag{$w_e \leq 2^i$ for all $e \in E_i$} \\
&= \frac{2^{2i}}{\alpha}\abs{E_i(u, \overline{S_{ij}})} \cdot \abs{E_i(u, S_{ij})} \\
&\leq \frac{2^{2i}}{\alpha}\abs{E_i(u,\overline{S_{ij}})}\cdot\abs{S_{ij}}. \tag{$\,\abs{E_i(u,S_{ij})} \le |S_{ij}|$}
\end{align*}
In Algorithm~\ref{alg:recovery-one-sketch}, we set $I_{V_{ij}} = \sum_{u \in S_{ij}} I_{S_{ij}}(u)$, so
\begin{align*}
\Var{I_{V_{ij}}} = \sum_{u \in S_{ij}}\Var{I_{S_{ij}}(u)}
&\leq \sum_{u \in S_{ij}}\frac{2^{2i}}{\alpha}\abs{S_{ij}}\cdot\abs{E_i(u,\overline{S_{ij}})} \\
&= \frac{2^{2i}}{\alpha}\cdot\abs{S_{ij}}\cdot X_{ij} \tag{$X_{ij} = \abs{E_i(S_{ij}, \overline{S_{ij}})}$} \\
&\leq \frac{2^{2i}}{\alpha \lambda}\left(X_{ij} + \overline{X_{ij}}\right)X_{ij}. \tag{$(S_{ij}, \overline{S_{ij}})$ is not $\lambda$-sparse}
\end{align*}
Summing across every dense component $V_{ij}$ in $G_i$, we get
\begin{align*}
\Var{I_{G_i}} = \sum_{j}\Var{I_{V_{ij}}}
&\leq \frac{2^{2i}}{\alpha \lambda}\sum_{j}\left(X_{ij} + \overline{X_{ij}}\right)X_{ij} \\
&\leq \frac{2^{2i}}{\alpha \lambda}\left(X_i + \overline X_i\right)X_i. \tag{$X_{i} = \sum_{j} X_{ij}$ and $\overline{X_{i}} = \sum_{j} \overline{X_{ij}}$}
\end{align*}
Finally, we sum across the weight classes indexed by $i$ to obtain
\begin{align*}
\Var{I_S} &= \sum_i \Var{I_{G_i}} \\
&\leq \frac{1}{\alpha\lambda} \sum_i \left(2^i\right)^2 \left(X_i + \overline X_i\right)X_i \\
&= \frac{1}{\alpha\lambda} \left[\sum_i \left(2^iX_i\right)^2 + \sum_i2^i\overline X_i \cdot\sum_i2^i X_i\right] \\
&\leq \frac{4}{\alpha\lambda} \left[w(S, \overline S)^2 + w(\overline S, S) \cdot w(S, \overline S) \right] \tag{$w_e \geq 2^{i-1}$ for all $e \in E_i$} \\
&\leq \frac{4}{\alpha\lambda} \left[w(S, \overline S)^2 + \beta w(S, \overline S)^2 \right] \tag{$ w(\overline S, S) \le \beta w(S, \overline S)$} \\
&= O\left(\frac{\beta}{\alpha\lambda}\right) w(S, \overline S)^2. \tag*{\qedhere}
\end{align*}
\end{proof}

\subsection{For-Each Cut Sketch: Faster Algorithms}
\label{sec:for-each-ub-faster}

In this section, we give nearly-linear time algorithms for computing and querying for-each cut sketches (when $\beta$ is known).
If $\beta$ is unknown, then we can first compute a constant approximation of it using an algorithm of Ene~{\em et al.}~\cite{ene2016routing} in $\tilde O(\beta^2 m)$ time. In other words, we can speed up the algorithms in the previous section and prove the following theorem.

\begin{theorem}
\label{thm:beta-12-faster}
Consider the same setting as in Theorem~\ref{thm:beta-12}.
That is, a $(1\pm\eps)$ for-each cut sketch of size $O(\beta^{1/2} n \log^{5}n / \eps)$ bits exists for any $n$-vertex $\beta$-balanced graph $G$.
Now in addition, we can compute such a cut sketch in time $\tilde O(m + \beta^{1/2} n / \eps)$.
\end{theorem}


At a high level, dense components are easier to sketch.
In order to decompose the graph into dense components, the previous sketching algorithm (Algorithm~\ref{alg:store}) iteratively finds and stores sparse cuts, which is a very slow process because sparsest cut is NP-Hard and we have to do this repeatedly.
We can speed this up by considering a more direct graph partitioning algorithms that decompose the graph into dense components, and that is expander partitioning.

The expander decomposition problem has been studied intensively (see, e.g.,~\cite{KannanVV04,SpielmanT04,KelnerLOS14,Wulff-Nilsen17,NanongkaiS17,SaranurakW19,ChuzhoyGLNPS19}), where the goal is to partition a graph into disjoint clusters, such that each cluster is internally well-connected while the number of cross-cluster edges is small.
Thus, expander decomposition has become a powerful algorithmic tool, especially in designing nearly-linear time algorithms for a wide range of fundamental graph and matrix problems. For our purposes, we use the following (randomized) subroutine from~\cite{SaranurakW19} that partitions a graph into expanders in nearly-linear time.

\begin{lemma}[Expander Decomposition, \cite{SaranurakW19}]
\label{lem:expander-p}
Given an undirected graph $G = (V, E)$ with $n = |V|$ and $m = |E|$, there is a randomized algorithm that with high probability finds a partitioning of $V$ into disjoint set of vertices $V_1, \ldots, V_k$ in time $\tilde O(m)$ such that
\begin{enumerate}
\item The $G[V_i]$'s contain at least half of the edges of $G$: $\sum_i |E(V_i, V_i)| \ge m/2$.
\item For every $i$, $G[V_i]$ has conductance $\Omega(1/\log^3 n)$.
\end{enumerate}
\end{lemma}


For undirected graphs, Jambulapati and Sidford~\cite{JambulapatiS18} showed how to construct for-each cut sketches in nearly-linear time.
Instead of trying to repeatedly find sparse cuts, they showed how to sketch expander graphs (graphs with high conductance) and then decompose the input graph using expander partitioning algorithms.

Intuitively, we should be able to speed up our algorithm using a similar approach, because when we partition the graph by removing sparse cuts, we do not look at the direction of the edges.
However, the analysis in~\cite{JambulapatiS18} does not apply to our setting because they focus on sketching quadratic forms.
The quadratic form of a directed Laplacian ignores edge directions and hence does not preserve the directed cut values.
On a more technical level, their analysis relies heavily on the notion of conductance, which is not canonically defined for directed graphs.
In our setting, we cannot bound the variance of our estimator even if we have a directed graph whose undirected version is an expander (see Appendix~\ref{sec:comp-js} for more details).

\paragraph*{Overview of the Faster Sketching and Recovery Algorithms.}

We now describe our nearly-linear time for-each cut sketching algorithm for balanced graphs (Algorithm~\ref{alg:faster-store}).
We first partition the edges into $\ell = O(\log n)$ weight classes $(E_i)_{i=1}^\ell$ and let $G_i = (V, E_i, w)$.
Now for every $G_i$, instead of iteratively finding sparse cuts as in Algorithm~\ref{alg:store}, we invoke Lemma~\ref{lem:expander-p} to obtain an expander decomposition, sketch the internal edges of the expanders via random sampling, remove them from $G_i$, and repeat this process on the remaining edges.

\begin{algorithm}[!ht]
\caption{Compute a $(1 \pm \eps)$ for-each cut sketch for $G$ in nearly-linear time.}
\label{alg:faster-store}
\SetKwInOut{Input}{Input}
\SetKwInOut{Output}{Output}
\setcounter{AlgoLine}{0}
\Input{A $\beta$-balanced graph $G = (V, E, w)$ with edge weights $w_e \in [1, \poly(n)]$ and $\beta \ge 1$, and $\epsilon \in (0, 1)$.}
\Output{A $(1 \pm \eps)$ for-each cut sketch of $G$.}
Let $\alpha = \frac{\beta^{1/2} \ln^{3/2} n}{\eps}$. \\
Partition the edges into $\ell = O(\log(n))$ weight classes $(E_i)_{i=1}^\ell$ where $E_i = \{e: w_e \in [2^{i-1}, 2^i) \}$. \\
Each weight class $E_i$ defines a (possibly unbalanced) graph $G_i = (V, E_i, w)$. \\
\For{$i = 1$ {\bf to} $\ell$}{
  Let $j = 1$, $G_{i1} = G_i$, and $E_{i1} = E_i$. \\
  \While{$E_{ij} \neq \varnothing$}{
    Compute a cut sketch $\sk(G_{ij})$ for $G_{ij}$ as follows: \\
    \begingroup
    \leftskip1em
    (a) Compute and store an expander decomposition $(V_{ijk})_k$ on the \emph{undirected, unweighted} version of $G_{ij}$ using Lemma~\ref{lem:expander-p}. \\
    (b) For every $V_{ijk}$, if any vertex $u \in V_{ijk}$ has $d^{out}_{ijk}(u) + d^{in}_{ijk}(u) \le \alpha$, then store and remove all edges incident to $u$. \label{alg:extra-step}
    Let $E_{ijk}$ denote the remaining edges in $G_{ij}[V_{ijk}]$. \\
    (c) For every $V_{ijk}$ and $u\in V_{ijk}$, store the in- and out-degree of $u$ in $V_{ijk}$, i.e., $d^{in}_{ijk}(u) = |E_{ijk}(V_{ijk},u)|$ and $d^{out}_{ijk}(u) = |E_{ijk}(u, V_{ijk})|$. \\
    (d) For every $V_{ijk}$, at each vertex $u \in V_{ijk}$, sample with replacement $\alpha$ edges from the outgoing edges at $u$ in $E_{ijk}$.
    Do the same at each vertex for incoming edges. \\
    \endgroup
    Copy the remaining cross-cluster edges to $G_{i,j+1}$, i.e., set $E_{i,j+1} = E_{ij} \setminus \bigcup_k E(V_{ijk}, V_{ijk})$, $G_{i,j+1} = (V, E_{i,j+1}, w)$, and $j = j + 1$;
  }
}
\Return{$\sk(G) = \bigcup_{i,j} \sk(G_{ij})$.}
\end{algorithm}

\begin{algorithm}[!ht]
\caption{Query the cut value $w(S, \overline S)$ from $\sk(G)$}
\label{alg:faster-query}
\SetKwInOut{Input}{Input}
\SetKwInOut{Output}{Output}
\setcounter{AlgoLine}{0}
\Input{A cut query $S \subseteq V$ and a cut sketch $\sk(G)$ (given by Algorithm~\ref{alg:faster-store}).}
\For{each $\sk(G_{ij})$ in $\sk(G)$}{
    \For{each expander $V_{ijk}$ in $G_{ij}$}{
        Let $S_{ijk}$ denote the smaller set of $(V_{ijk} \cap S)$ and $(V_{ijk} \cap \overline S)$. \\
        \uIf{$S_{ijk} = V_{ijk} \cap S$}{
            Estimate the total weight of edges leaving $S_{ijk}$:
            For every $u \in S_{ijk}$, set
            \begin{equation} \label{eq:out-est-at-u-2}
            I_{S_{ijk}}(u) = \frac{d^{out}_{ijk}(u)}{\alpha} \sum_{q=1}^\alpha \chi_{ijk}(u,q)w_{ijk}(u,q)
            \end{equation}
            where $d^{out}_{ijk}(u)$ is the out-degree of $u$ in $V_{ijk}$, $\chi_{ijk}(u,q) = 1$ if the $q$-th sampled outgoing edge at $u$ crosses $S$ and $\chi_{ijk}(u,q) = 0$ otherwise, and $w_{ijk}(u,q)$ is the weight of the $q$-th sampled edge.
        }
        \Else{
            Estimate the total weight of edges entering $S_{ijk} = V_{ijk} \cap \overline S$ instead: \\
            For every $u \in S_{ijk}$, set $I_{S_{ijk}}(u)$ as in \eqref{eq:out-est-at-u-2}, using $d^{in}_{ijk}(u)$ instead of $d^{out}_{ijk}(u)$, and $\chi_{ijk}(u,q)$ indicates if the $q$-th sampled incoming edge at $u$ crosses $S$.
        }
        The estimated contribution from $V_{ijk}$ is $I_{V_{ijk}} = \sum_{u \in S_{ijk}} I_{S_{ijk}}(u)$.
    }
}
Compute the overall estimate of total weight of dense-component edges as $I_S = \sum_{i,j,k} I_{V_{ijk}}$. \\
Compute $J_S$, the total weight of the edges stored in Step~\ref{alg:extra-step}(b) of Algorithm~\ref{alg:faster-store} that leaves $S$. \\
\Return{$I_S + J_S$.}
\end{algorithm}

Formally, we let $G_{ij}$ (and $G_{i,j}$) denote the remainder of $G_i$ after $(j-1)$ iterations; the edge set is $E_{ij}$ (and $E_{i,j}$). Initially, we have $G_{i1} = G_i$.
We compute an expander decomposition of $G_{ij}$, and then compute a cut sketch $\sk(G_{ij})$ for the edges inside the expanders.
We copy the cross-cluster edges to $G_{i,j+1}$ and sketch them later.
Observe that $E_{i,j+1}$ is a subset of $E_{i,j}$, and by Lemma~\ref{lem:expander-p}, the number of edges in $E_{i,j+1}$ is at most half of that of $E_{i,j}$.
This guarantees that after $j = O(\log n)$ iterations, $G_{ij}$ must be empty, which allows us to bound the size of $\sk(G)$ and the running time of Algorithm~\ref{alg:faster-store}.

One change in Algorithm~\ref{alg:faster-store} is that we store all edges incident to low-degree vertices (as proposed in~\cite{JambulapatiS18}).
Let $\{V_{ijk}\}_k$ be an expander decomposition of $G_{ij}$.
For every $V_{ijk}$, if there are at most $\alpha$ (incoming and outgoing) edges at $u \in V_{ijk}$, then we store and remove these edges.
Let $E_{ijk}$ denote the remaining edges in $G_{ij}[V_{ijk}]$.
Finally, for every node $u \in V_{ijk}$, we store its incoming and outgoing degrees in $E_{ijk}$, and sample uniformly at random (with replacement) $\alpha$ incoming and outgoing edges at $u$ within $E_{ijk}$.

In Algorithm~\ref{alg:faster-query}, the new query algorithm, we estimate the cut value $w(S,\overline S)$ by summing over every $G_i$ (defined by edge weights), every $G_{ij}$ (created during recursive expander decomposition), every expander $V_{ijk}$ in $G_{ij}$, and finally every node in $S \cap V_{ijk}$.
Algorithm~\ref{alg:faster-query} has two important changes compared to the previous query algorithm (Algorithm~\ref{alg:recovery-one-sketch}):
(1) we no longer have sparse-cut edges, so $J_S$ is now an estimator for edges incident to low-degree vertices, and (2) we have one more index due to recursive expander decomposition and the summation over all expanders $V_{ijk}$.

\paragraph*{Outline of the Rest of the Section.} In the rest of this section we prove three key lemmas.
Lemma~\ref{lem:faster-size} bounds the size of $\sk(G)$ outputted by the new sketching algorithm,
Lemma~\ref{lem:faster-runtime} bounds its running time,
and finally, Lemma~\ref{lem:faster-var} proves the correctness of the new query algorithm.
Theorem~\ref{thm:beta-12-faster} follows immediately from these three lemmas.

\medskip
\begin{lemma}[Output Size of Algorithm~\ref{alg:faster-store}]
\label{lem:faster-size}
Let $G = (V, E, w)$ be a $\beta$-balanced directed graph with $n = |V|$ and $w_e \in [1, \poly(n)]$.
On input $(G, \beta, \eps)$, Algorithm~\ref{alg:faster-store} outputs a data structure $\sk(G)$ of size $O(n \beta^{1/2} \log^5 n / \eps)$.
\end{lemma}
\begin{proof}
We have $\ell = O(\log n)$ graphs $G_1, \ldots, G_\ell$, one for each edge-weight class.
For every $G_i$, we iteratively perform expander partitioning and obtain $(G_{ij})_j$.
There are at most $O(\log n)$ graphs $G_{ij}$ for every $G_i$, because $|E_{ij}| \le n^2$ and we have $|E_{i,j+1}| \le |E_{i,j}| / 2$ by Lemma~\ref{lem:expander-p}.

Therefore, it is sufficient to bound the size of each $\sk(G_{ij})$.

In every $\sk(G_{ij})$:
\begin{itemize}
    \item The expanders $V_{ijk}$'s are disjoint, so we can store the vertices they have in $O(n \log n)$ bits.
    \item For every vertex $u \in \bigcup_k V_{ijk}$, if the number of incoming and outgoing edges at $u$ is at most $\alpha$ then we store and remove all of them, which takes $O(\alpha n \log n)$ bits.
    \item For every vertex $u \in \bigcup_k V_{ijk}$, we store the in- and out-degree of $u$ in $E_{ijk}$, which takes $O(n \log n)$ bits.
    \item For every vertex $u \in \bigcup_k V_{ijk}$ that have more than $\alpha$ edges, we sample and store $\alpha$ incoming and outgoing edges at $u$.
    They can be stored using $O(\alpha n \log n)$ bits.
\end{itemize}
In summary, the overall size of $\sk(G)$ is
\[
O(\log n) \cdot O(\log n) \cdot O(n \log n + \alpha n \log n) = O(n \alpha \log^3 n) = O(n \beta^{1/2} \log^{9/2} n / \eps).
\]
The last step is because we choose $\alpha = \frac{\beta^{1/2} \ln^{3/2} n}{\eps}$ in Algorithm~\ref{alg:faster-store}.
\end{proof}

\begin{lemma}[Running Time of Algorithm~\ref{alg:faster-store}]
\label{lem:faster-runtime}
Let $G = (V, E, w)$ be a $\beta$-balanced directed graph with $n = |V|$ and $w_e \in [1, \poly(n)]$.
On input $(G, \beta, \eps)$, Algorithm~\ref{alg:faster-store} runs in time $\tilde O(m + n \alpha) = \tilde O(m + n \beta^{1/2} / \eps)$.
\end{lemma}
\begin{proof}
Recall that we have $O(\log n)$ graphs $(G_i)_i$, one for each weight class. Computing all the $G_i$'s can be done in $\tilde O(m)$ time by simply checking the weight of every edge.

For each $G_i$, we construct $O(\log n)$ graphs $(G_{ij})_j$.
For every $G_{ij}$, we first compute an expander decomposition using Lemma~\ref{lem:expander-p}, which runs in time $\tilde O(m)$.
The cross-cluster edges are sketched later and they can be copied to $G_{i,j+1}$ in time $O(m)$.

For each expander $V_{ijk}$ in $G_{ij}$, we can check which vertices have degree at most $\alpha$, store and remove all edges incident to these vertices in time $\tilde O(\alpha n)$.
Storing the in- and out-degree of every vertex in $(V_{ijk}, E_{ijk})$ takes $\tilde O(n)$ time, and storing $2 \alpha$ sampled incoming and outgoing edges from each vertex requires $\tilde O(n\alpha)$ time.

In summary, the overall running time of Algorithm~\ref{alg:faster-store} is
\[
O(\log n) \cdot \left(\tilde O(m) + O(\log n) \cdot \tilde O(m + n \alpha)\right) = \tilde O(m + n \alpha) = \tilde O(m + n \beta^{1/2} / \eps). \qedhere
\]
\end{proof}

We make some without-loss-of-generality assumptions and introduce some new notation before proving the correctness of Algorithm~\ref{alg:faster-query}.

Fix a cut query $S \subset V$.
Recall that $(V_{ijk})_k$ is an expander partitioning of $G_{ij}$ given by Lemma~\ref{lem:expander-p}.
Without loss of generality, we can assume $|V_{ijk} \cap S| \le |V_{ijk} \cap \overline{S}|$ for every $V_{ijk}$. Otherwise, Algorithm~\ref{alg:faster-query} works with $V_{ijk} \cap \overline{S}$ and queries for the incoming edges instead.
Under this assumption, we always have $S_{ijk} = V_{ijk} \cap S$.
We use $\overline{S_{ijk}}$ to denote $V_{ijk}\setminus S_{ijk}$.

We further assume without loss generality that $J_S = 0$, i.e., there are no low-degree vertices in every $V_{ijk}$.
This is because $J_S$ is unbiased and deterministic, and we want to prove
\[
\Var{I_S} \le O(\eps^2) \cdot w(S,\overline S)^2 = O(\eps^2) \cdot(\Ex{I_S} + J_S)^2.
\]
This is harder to prove if we set $J_S = 0$.
Under this assumption, $E_{ijk} = E_{ij}(V_{ijk}, V_{ijk})$.

\begin{lemma}[Correctness of Algorithm~\ref{alg:faster-query}]
\label{lem:faster-var}
Let $\sk(G)$ be the output of Algorithm~\ref{alg:faster-store}.
Fix a cut query $S \subseteq V$. With probability at least $2/3$, the estimate $(I_S+J_S)$ returned by Algorithm~\ref{alg:faster-query} on input $(S, \sk(G))$ satisfies
\[
\abs{\left(I_S + J_S\right) - w(S, \overline S)} \le O(\eps) \cdot w(S, \overline S).
\]
\end{lemma}
\begin{proof}
We will show that $I_S$ is an unbiased estimator and
\[
\Var{I_S} \le O(\eps^2) \cdot w(S, \overline S).
\]
The lemma then follows from Chebyshev's inequality.

By our assumption we always work with outgoing edges.
Recall that our estimator is
\[
I_S = \sum_i \sum_j \sum_k \sum_{u \in S_{ijk}} I_{S_{ijk}}(u),
\]
where
\[
I_{S_{ijk}}(u) = \frac{d^{out}_{ijk}(u)}{\alpha}\sum_{q=1}^\alpha \chi_{ijk}(u,q) w_{ijk}(u,q).
\]

We first show that $I_S$ is unbiased.
Every edge $e \in E$ of $G$ is counted exactly once in some $V_{ijk}$;
more specifically, every edge $e$ belongs to exactly one $G_i$, where it is sketched in some expander $V_{ijk}$ in some $G_{ij}$ and then removed.
Therefore, it is sufficient to prove that $I_{S_{ijk}}(u)$ is unbiased.
At $u \in V_{ijk}$, there are $d^{out}_{ijk}(u)$ edges and we sample $\alpha$ edges independently, so if we scale the edge weights by $(d^{out}_{ijk} / \alpha)$ and sum over the sampled edges that cross $S$, we get an unbiased estimator. 


Recall that $G_{ij}[V_{ijk}]$ has conductance $\Omega(\frac{1}{\log^3 n})$ and $S_{ijk} = V_{ijk} \cap S$.
Let $X_{ijk} = |E_{ij}(S_{ijk}, \overline{S_{ijk}})|$ denote the number of edges that go from $S_{ijk}$ to $\overline{S_{ijk}}$ in $G_{ij}$, and $\overline{X_{ijk}}$ the number of edges in the reverse direction.
Let $X_{ij} = \sum_k X_{ijk}$ and $X_i = \sum_j X_{ij}$.
Observe that $X_i$ is precisely the total number of edges that goes from $S$ to $\overline{S}$ in $G_i$ because every edge appears exactly once in some $V_{ijk}$.

Our proof here shares the same structure as the proof of Lemma~\ref{lem:var-s}.
We first outline some of the most significant changes compared to our previous proof.

\medskip
{\noindent \bf Storing All Edges of Low-Degree Vertices.} In the induced subgraph $G_{ij}[V_{ijk}]$, if the total incoming and outgoing degree of a vertex $u \in V_{ijk}$ is at most $\alpha$, then we store all these edges.
These edges form the deterministic estimator $J_S$.

Recall that without loss of generality, we can assume $J_S = 0$, that is, there are no low-degree in any $V_{ijk}$.
Thus, for any $u \in V_{ijk}$, we have
\[
d^{out}_{ijk}(u) + 
d^{in}_{ijk}(u) \ge \alpha.
\]
Summing over all vertices $u \in S_{ijk}$ (or $\overline{S_{ijk}}$), we have
\begin{equation}
\alpha \cdot |S_{ijk}| \le |E(S_{ijk}, V_{ijk})| + 
|E(V_{ijk}, S_{ijk})|, \quad \alpha \cdot |\overline{S_{ijk}}| \le |E(\overline{S_{ijk}}, V_{ijk})| + 
|E(V_{ijk}, \overline{S_{ijk}})|
\label{eqn:sv-vs-ge}
\end{equation}

\medskip
{\noindent \bf Using Conductance Rather Than Sparse Cuts.} By the definition the conductance (see Equation~\ref{eqn:conductance}), because the conductance of the undirected unweighted version of $G_{ij}[V_{ijk}]$ is $\Omega(1/\log^3 n)$, for any disjoint partition $(S_{ijk}, \overline{S_{ijk}})$ of $V_{ijk}$,
\[
\frac{ X_{ijk} + \overline{X_{ijk}} }{
  \min\left(|E_{ij}(S_{ijk},S_{ijk})|, |E_{ij}(\overline{S_{ijk}},\overline{S_{ijk}})|\right)
    + X_{ijk} + \overline{X_{ijk}} }
\ge \phi(G_{ij}[V_{ijk}]) = \Omega\left(\frac{1}{\log^3 n}\right).
\]
Consequently, because $|E_{ij}(S_{ijk}, V_{ijk})| = |E_{ij}(S_{ijk}, S_{ijk})| + X_{ijk}$ and so on, we have
\begin{align}
\begin{split}
&\quad \min\left(|E_{ij}(S_{ijk}, V_{ijk})| + 
|E_{ij}(V_{ijk}, S_{ijk})|,
|E_{ij}(\overline{S_{ijk}}, V_{ijk})| + 
|E_{ij}(V_{ijk}, \overline{S_{ijk}})|\right) \\
&= 2\min\left(|E_{ij}(S_{ijk},S_{ijk})|, |E_{ij}(\overline{S_{ijk}},\overline{S_{ijk}})|\right) + X_{ijk} + \overline{X_{ijk}} \\
&\le O(\log^3 n) \cdot (X_{ijk} + \overline{X_{ijk}}).
\end{split}
\label{eqn:sv-vs-le}
\end{align}
Recall that we can assume $|S_{ijk}| \le |\overline{S_{ijk}}|$.
Combining Inequalities~\eqref{eqn:sv-vs-le} and~\eqref{eqn:sv-vs-ge}, we have
\begin{align}
\begin{split}
\alpha \cdot |S_{ijk}|
&= \alpha \cdot \min(|S_{ijk}|, |\overline{S_{ijk}}|) \\
&\le \min\left(|E_{ij}(S_{ijk}, V_{ijk})| + 
|E_{ij}(V_{ijk}, S_{ijk})|,
|E_{ij}(\overline{S_{ijk}}, V_{ijk})| + 
|E_{ij}(V_{ijk}, \overline{S_{ijk}})|\right) \\
&\le O(\log^3 n) \cdot (X_{ijk} + \overline{X_{ijk}}).
\end{split}
\label{eqn:sijk-le}
\end{align}

We will use Inequality~\eqref{eqn:sijk-le} to relate the variance of $I_{S_{ijk}} = \sum_{u\in S_{ijk}}I_{S_{ijk}}(u)$ with $\left(X_{ijk} + \overline{X_{ijk}}\right)$.
Consequently, the variance of $I_S$ is related to $X_i = \sum_{j} \sum_{k} X_{ijk}$.

\medskip
{\noindent \bf Upper Bounding the Variance.} We now formally upper bound the variance of $I_S$.

Fix any $u \in S_{ijk}$.
We first upper bound $\Var{I_{S_{ijk}}(u)}$.
By definition, $\chi_{ijk}(u,q)=1$ if the $q$-th sampled edge leaving $u$ goes to $\overline S$ (and $\chi_{ijk}(u,q)=0$ otherwise), so
\[
\Var{\chi_{ijk}(u,q)} = \frac{\abs{E_{ij}(u, \overline{S_{ij}})}}{d^{out}_{ijk}(u)} \cdot \frac{\abs{E_{ij}(u, S_{ijk})}}{d^{out}_{ijk}(u)}.
\]
Now from the definition of $I_{S_{ijk}}(u)$, we have
\begin{align*}
\Var{I_{S_{ijk}}(u)} &= \frac{d^{out}_{ijk}(u)^2}{\alpha^2}\sum_{q=1}^\alpha \Var{\chi_{ijk}(u,q)} w_{ijk}(u,q)^2 \\
&\leq \frac{d^{out}_{ijk}(u)^2}{\alpha^2} \cdot \alpha \cdot \frac{\abs{E_{ij}(u, \overline{S_{ijk}})}}{d^{out}_{ijk}(u)} \cdot \frac{\abs{E_{ij}(u, S_{ijk})}}{d^{out}_{ijk}(u)} \cdot 2^{2i} \tag{$w_e \le 2^i$ for all $e \in E_{ij}$} \\
&= \frac{2^{2i}}{\alpha} \abs{E_{ij}(u, \overline{S_{ijk}})} \cdot \abs{E_{ij}(u, S_{ijk})} \\
&\le \frac{2^{2i}}{\alpha} \abs{E_{ij}(u,\overline{S_{ijk}})}\cdot\abs{S_{ijk}}. \tag{$\,\abs{E(u,S_{ijk})} \le |S_{ijk}|$}
\end{align*}
Let $I_{V_{ijk}} = \sum_{u \in S_{ijk}} I_{S_{ijk}}(u)$.
Summing across every vertex $u \in S_{ijk}$, we get
\begin{align*}
\Var{I_{V_{ijk}}}
= \sum_{u \in S_{ijk}}\Var{I_{S_{ijk}}(u)} 
&\leq \sum_{u \in S_{ijk}} \frac{2^{2i}}{\alpha} \abs{S_{ijk}}\cdot\abs{E(u, \overline{S_{ijk}})} \\
&= \frac{2^{2i}}{\alpha} \abs{S_{ijk}}\cdot X_{ijk} \tag{$X_{ijk} = \abs{E(S_{ijk}, \overline{S_{ijk}})}$} \\
&\leq O\left(\frac{2^{2i} \log^3 n}{\alpha^2}\right) \left(X_{ijk} + \overline{X_{ijk}}\right) \cdot X_{ijk}. \qquad \quad \tag{Inequality~\eqref{eqn:sijk-le}}
\end{align*}

The rest of the proof is almost identical to the proof of Lemma~\ref{lem:var-s}, so we omit some details.

We sum over every expander $V_{ijk}$ in $G_{ij}$, then over the graphs $G_{ij}$ obtained from recursive expander partitioning, and finally over the weight classes indexed by $i$.
Let $I_{S} = \sum_{i} \sum_{j} \sum_{k} I_{V_{ijk}}$.
Using the fact that $X_{ij} = \sum_{k} X_{ijk}$ and $X_i = \sum_{j} X_{ij}$ (and similarly for $\overline{X_{ij}}$ and $\overline{X_i}$), we have
\begin{align*}
\Var{I_S}
  &\le O\left(\frac{2^{2i} \log^3 n}{\alpha^2}\right) \sum_i \left(X_i + \overline{X_i} \right) \cdot X_i \\
  &\le O\left(\frac{\log^3 n}{\alpha^2}\right) \left[w(S, \overline S)^2 + w(\overline S, S) \cdot w(S, \overline S) \right] \tag{$w_e \geq 2^{i-1}$ for all $e \in E_i$} \\
  &\le O\left(\frac{\log^3 n}{\alpha^2}\right) \left[w(S, \overline S)^2 + \beta w(S, \overline S)^2 \right] \tag{$ w(\overline S, S) \le \beta w(S, \overline S)$} \\
  &= O\left(\frac{\beta \log^3 n}{\alpha^2}\right) w(S, \overline S)^2. 
\end{align*}
As we choose $\alpha = \frac{\beta^{1/2} \ln^{3/2} n}{\eps}$ in Algorithm~\ref{alg:faster-store}, we have $\Var{I_S} \le O(\eps^2) \cdot w(S,\overline S)^2$ as needed.
\end{proof}

\section{For-Each Cut Sketch: $\Omega(n\cdot \sqrt{\beta/\epsilon})$ Lower Bound} \label{sec:for-each-lb}


In this section, we prove that the size of for-each cut sketches must scale with $\sqrt{\beta}$.

\begin{theorem}
\label{thm:anycut-lb}
Fix $\beta \ge 1$ and $0 < \eps < 1$ with $\left(\beta/\eps\right)^{1/2} \le \frac{n}{2}$.
A $(1\pm\eps)$ for-each cut sketching algorithm for $n$-node $\beta$-balanced graphs must output $\Omega(n \cdot (\beta / \eps)^{1/2})$ bits in the worst case.
\end{theorem}

To prove this, we will need the following folklore result from communication complexity:
\begin{lemma}
\label{lem:cc-everybit}
Given a bit string $s \in \{0, 1\}^N$, if there is a data structure $D$ that allows one to recover each bit of $s$ with marginal probability at least $2/3$, then $D$ must use $\Omega(N)$ bits.
\end{lemma}

We first prove a special case of our lower bound for specific values of $\beta = \Theta(n^2)$ and $\eps = \Theta(1)$ (Lemma~\ref{lem:anycut-lb-simple}).
The proof for this special case is easier to explain and it contains the key ingredients of our construction for the general lower bound.

\begin{lemma}
\label{lem:anycut-lb-simple}
For $\beta = n^2$ and $\eps = \frac{1}{10}$, any $(1\pm\eps)$ for-each cut sketching algorithm for $n$-node $\beta$-balanced graphs must output $\Omega(n \cdot \beta^{1/2})$ bits in the worst case.
\end{lemma}

\begin{wrapfigure}{r}{0.35\textwidth}
\centering
\includegraphics[width=.26\textwidth]{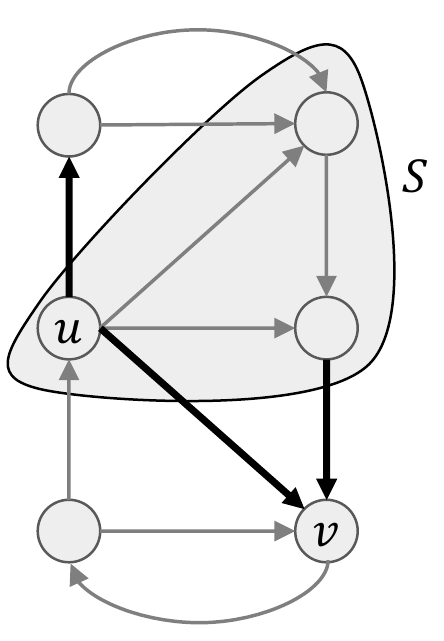}
\caption{In this example, the cut value $w(S,\overline S)$ is 2 or 3, depending on the value of $w(u,v)$. Thus, a $(1\pm 0.1)$-approximation to $w(S,\overline S)$ allows us to decode the bit in $s$ corresponding to edge $(u,v)$. (For readability we omit other bipartite edges from $L$ to $R$.)}
\label{fig:sqrt-lb}
\vspace{-0.5cm}
\end{wrapfigure}

\begin{proof}
At a high level, we will encode a bit string $s$ of length $\Omega(n^2)$ into an $n$-node $\beta$-balanced graph, such that given a $(1 \pm \eps)$ for-each cut sketch, we can recover each bit of $s$ with high constant probability.
Then by Lemma~\ref{lem:cc-everybit}, the cut sketch must have use $\Omega(|s|) = \Omega(n^2) = \Omega(n \beta^{1/2})$ bits.

Given a bit string $s$ of length $\frac{n^2}{4}$, we construct a graph as follows.
We start with an $\frac{n}{2} \times \frac{n}{2}$ complete bipartite digraph where edges go from left to right.
We set the weight of the $i$-th bipartite edge to $s_i + 1$ (so either $1$ or $2$).
We add a unit-weight cycle that leaves each side exactly once.
See Fig.~\ref{fig:sqrt-lb} for an example of our construction.

We first show that the graph is $\beta$-balanced for $\beta = n^2$.
The graph is strongly connected because it contains a cycle.
Note that all edge weights are in $[1, 2]$ and there are in total $\frac{n^2}{4} + n \le \frac{n^2}{2}$ edges in the graph.
Therefore, for every non-empty set $S \subset V$, the total weight of edges leaving (or entering) $S$ is at least $1$ and at most $n^2$, so the graph is $(n^2)$-balanced.

It remains to show that we can recover each bit of $s$ from a $(1\pm\eps)$ cut sketch.
Let $L$ denote the left vertices and $R$ the right vertices.
Fix any coordinate of $s$ and suppose it corresponds to the edge $(u, v)$ for some $u \in L$ and $v \in R$.
To recover this bit of $s$, we need to decide whether $w(u,v)$ is $1$ or $2$.
Consider the cut value leaving $S = \{u\} \cup R \setminus \{v\}$.
The cycle contributes a fixed amount to this cut (independent of the weights of the bipartite edges), which is at most $3$.
More importantly, $(u, v)$ is the only bipartite edge leaving $S$.
Since $\eps = \frac{1}{10}$ and the sketch returns this cut value within a factor of $(1 \pm \eps)$ with probability at least $2/3$, we can recover $w(u,v)$ with probability at least $2/3$.
\end{proof}

Our lower bound construction for general values of $\beta$ and $\eps$ builds on the one in the proof of Lemma~\ref{lem:anycut-lb-simple}.
At a high level, instead of using a bipartite graph with two clusters, we will use multiple clusters where the size of each cluster depends on $\beta$ and $\eps$.

\begin{proof}[Proof of Theorem~\ref{thm:anycut-lb}]
Let $k = \sqrt{\beta/\eps}$.
We will encode a bit string $s$ of length $\Omega(nk)$ into an $n$-node graph $G$ such that (1) $G$ is $(3 \beta$)-balanced, and (2) we can recover each bit of $s$ with high constant probability given a $(1 \pm c\cdot \eps)$ for-each cut sketch of $G$ where $c = 10^{-2}$.
By Lemma~\ref{lem:cc-everybit}, the cut sketch must have at least $\Omega(n k) = \Omega(n \cdot (\beta / \eps)^{1/2})$ bits.

Without loss of generality, we assume $k$ is an integer and $n$ is a multiple of $k$. We partition $n$ vertices into $t = n/k$ clusters of size $k$, which we denote by $V_1, \ldots, V_t$. Since we assume $n \ge 2 k$, there are at least two clusters.

Let $s$ be a bit string of length $k^2 (t-1) = \Omega(n k)$.
We partition $s$ into $(t-1)$ blocks where each block has length $k^2$.
We encode the $i$-th block of $s$ in a $k \times k$ complete bipartite digraph where edges go from $V_i$ to $V_{i+1}$.
As in the proof of Lemma~\ref{lem:anycut-lb-simple}, a bipartite edge $(u, v)$ for $u \in V_i$ and $v \in V_{i+1}$ has weight $s_{i,(u,v)} + 1$ (so either $1$ or $2$).
For every $1 \le i \le t-1$, we add a cycle between $V_i$ and $V_{i+1}$ that leaves $V_i$ and $V_{i+1}$ exactly once.
Now, in contrast to the previous construction, these cycle edges have weight $1/\eps$.

We first show that $G$ is $(3 \beta)$-balanced.
Fix any non-empty set $S \subset V$.
Let $G_i$ denote the subgraph between $V_i$ and $V_{i+1}$ which contains $k^2$ bipartite edges and one cycle.
Let $w_i(S, \overline{S})$ denote the total weight of edges leaving $S$ in $G_i$.
We will show that $w_i(S, \overline{S})$ and $w_i(S, \overline{S})$ are within a factor of $3 \beta$ of each other.
Because $G$ is strongly connected and $w(S, \overline{S}) = \sum_{i=1}^{t-1} w_i(S, \overline{S})$, we can conclude that $G$ is $(3 \beta)$-balanced.

Without loss of generality, we assume both $w_i(S, \overline{S})$ and $w_i(\overline{S}, S)$ are positive.
The cut value $w_i(S, \overline{S})$ remains the same if we restrict $G_i$ on vertices $(V_{i} \cup V_{i+1})$ and consider the cut query $S \cap (V_i \cup V_{i+1})$.
The cycle contributes equally in both directions, so without loss of generality, we can assume the cycle has minimum contribution, which is $\frac{1}{\eps}$.
(If the cycle contributes more, the cut is more balanced.)
The total weight of the bipartite edges is at most $2 k^2 = \frac{2 \beta}{\eps}$.
Therefore, the ratio between the cut values in both directions is at most $\frac{(2 \beta/\eps) + (1/\eps)}{1/\eps} = 2 \beta + 1 \le 3 \beta$.

It remains to show that we can recover every bit of $s$ from a cut sketch.
Fix any bit of $s$.
Suppose this bit $s_{i,(u,v)}$ corresponds to the edge $(u, v)$ for some $u \in V_i$ and $v \in V_{i+1}$, we query the cut value leaving $
S_{(u,v)} = \{u\} \cup \left(V_{i+1} \setminus \{v\}\right) \bigcup_{j={i+2}}^{t-1} V_j$.
The only bipartite edge leaving $S_{(u,v)}$ is the edge $(u, v)$, which has weight either $1$ or $2$.
There are at most $5$ cycle edges leaving $S_{(u,v)}$ (at most $1$ from $G_{i-1}$, $3$ from $G_i$, and $1$ from $G_{i+1}$), whose total weight is fixed and at most $\frac{5}{\eps}$.
Therefore, if we can compute an $(1 \pm c\cdot \eps)$ approximation to the cut value for $c = 10^{-2}$, we can recover the corresponding bit of $s$.
\end{proof}

\section{Conclusion}

In this paper, we considered the question of sparsifying directed graphs.
We focused on graphs that are $\beta$-balanced, where the ratio between the cut value in two directions is at most $\beta$.
We gave upper and lower bounds on the size of the cut sketch with almost tight dependence on $\beta$, under both the standard ``for-all'' notion (i.e., simultaneously preserving the value of all cuts) and the ``for-each'' notion (introduced by Andoni {\em et. al}~\cite{andoni2016sketching}) of cut sparsification.
More specifically, we showed that under the ``for-all'' notion, the linear dependence on $\beta$ obtained by Ikeda and Tanigawa~\cite{ikeda2018cut} is tight.
For the ``for-each'' notion, we gave a data structure that preserves cut values whose size scales as $\sqrt{\beta}$, thereby beating the ``for-all'' lower bound.
We also showed that this dependence on $\sqrt{\beta}$ is tight.
Our lower bounds hold not only for sparsifiers (i.e., graph encodings), but also for arbitrary data structures.

An interesting direction for future work is to consider the {\em spectral} sparsification of directed graphs. Cohen {\em et al.}~\cite{CohenKPPRSV17,CohenKPPRSV18} (see also Chu {\em et al.}~\cite{ChuGPSSW18}) introduced a novel definition of directed sparsification and leveraged it to solve directed Laplacian linear systems.
However, their work is not immediately relevant to ours because their directed spectral sparsifiers do not necessarily preserve directed cut values.
This motivates the following natural question: {\em is there a notion of spectral sparsification that generalizes cut sparsification in directed graphs}? (Note that this is indeed the case for undirected graphs, where spectral sparsifiers also preserve cut values.)
A natural candidate would be a sparse graph that preserves $\sum_{(u,v) \in E} \left((x_u - x_v)^+\right)^2$ for all real vectors $x$, where $y^+ = \max(0, y)$. Note that if $x\in \{0, 1\}^{|V|}$, then this sum represents directed cut values, which is analogous to the correspondence between cut and spectral sparsification in undirected graphs. It would be interesting to explore if preserving this sum in directed graphs has interesting applications beyond preserving cuts, and if so, whether there exist sparse graphs that preserve this sum approximately for balanced directed graphs.


\bibliographystyle{abbrv}
\bibliography{bibliography}

\clearpage

\appendix
\section{Technical Comparisons with Previous Work}
\subsection{Comparison with Andoni et al.~\cite{andoni2016sketching}} \label{sec:comp-andoni}

We compare (the undirected version of) our analysis with the one in~\cite{andoni2016sketching} at a more technical level.
For simplicity, suppose the input graph $G$ is unweighted and there are no $\lambda$-sparse cuts in $G$.
In this case, we want to approximate the directed cut value $X = |E(S, V \setminus S)|$.

Both sketching algorithms sample $\alpha$ edges at every vertex $u \in S$ for some $\alpha$. Thus, at each vertex $u \in S$, we have $\alpha$ independent and identically distributed Bernoulli variables, and each is 1 with probability $p =\frac{|E(u,S)|}{|E(u,V)|}$.
The variance of this random variable is $p(1-p)$.
Andoni et al.~\cite{andoni2016sketching} upper bounded this quantity by $(1-p)$, which resulted in the following bound on $\Var{I_S}$:
\begin{align*}
\Var{I_S} &\le \sum_{u \in S} \frac{\bigl(d(u)\bigr)^2}{\alpha^2} \cdot \alpha \cdot \frac{|S|}{d(u)} \\
 &= \frac{1}{\alpha}\sum_{u \in S} |S| \cdot d(u)
 = \frac{1}{\alpha} |S| \cdot |E(S, V)|
 \le \frac{1}{\alpha} |S| \cdot(|S|^2 + X) = \frac{1}{\alpha} \left(|S|^3 + |S| \cdot X\right).
\end{align*}

Our variance bound is obtained by calculating the variance of the estimator more carefully.
If we use the exact value of $\Var{I_S(u)}$, we have
\begin{align*}
\Var{I_S} &\le 
\sum_{u \in S} \frac{\bigl(d(u)\bigr)^2}{\alpha^2} \cdot \alpha \cdot \frac{|E(u, S)|}{|E(u, V)|} \frac{|E(u, \overline S)|}{|E(u, V)|} \\
&= \frac{1}{\alpha}\sum_{u \in S} |E(u, S)| \cdot |E(u, \overline S)| 
\le \frac{1}{\alpha} \left(\max_u |E(u, S)|\right) \left(\sum_u |E(u, \overline S)|\right) = \frac{1}{\alpha} |S| \cdot X.
\end{align*}

One of our main technical contributions is to remove the $|S|^3$ term from the upper bound on $\Var{I_S}$.
There are two consequences:
\begin{itemize}
\item
In the undirected setting, the redundant $|S|^3$ term in their analysis is precisely the reason why they need down-sampling (and consequently a constant-approximate sketch to choose the right down-sampling rate).
By down-sampling, they can guarantee that $X = O(1 / \eps^2)$ and $|S| \le \frac{X}{\lambda} = O(1/\eps)$ and hence $|S|^3$ is comparable to $|S|\cdot X$.
Consequently, our algorithm does not need down-sampling for both directed and undirected graphs.

\item
In the directed setting, the redundant $|S|^3$ term causes more severe issues.
Even with down-sampling, we can only guarantee that $X = O(1/\eps^2)$ and $|S| \le \frac{X + \overline X}{\lambda} \le \frac{\beta}{\lambda} X$.
If we set $\lambda$ to be as large as we can afford, i.e., $\lambda = \Theta(\sqrt{\beta})$, we still have $|S|^3 \gg |S| \cdot X$ and therefore removing this term is crucial for us to obtain a space complexity with the right dependence on $\beta$.
\end{itemize}

\subsection{Comparison with Jambulapati and Sidford~\cite{JambulapatiS18}}
\label{sec:comp-js}

Jambulapati and Sidford~\cite{JambulapatiS18} constructed graph sketches to approximately preserve quadratic forms.
As an important special case, their sketch can be used to query cut values.
In this section, we compare our proof with 
(the cut version of) their analysis.

For simplicity, suppose the input graph $G$ is unweighted and the conductance of $G$ is very high $\phi(G) = \Omega(1)$.
We want to approximate the directed cut value $X = |E(S, V \setminus S)|$.
Similar to Andoni et al.~\cite{andoni2016sketching}, Jambulapati and Sidford~\cite{JambulapatiS18} upper bounded $\Var{I_S}$ starting with
\begin{align*}
\Var{I_S} &\le \sum_{u \in S} \frac{\bigl(d(u)\bigr)^2}{\alpha^2} \cdot \alpha \cdot \frac{|S|}{d(u)}
 = \frac{1}{\alpha}\sum_{u \in S} |S| \cdot d(u).
\end{align*}

Because they store all edges incident to low-degree vertices, without loss of generality, one can assume $d(v) \ge \alpha$ for all $v \in V$.
Thus,
\begin{align*}
\Var{I_S} \le \frac{1}{\alpha}\sum_{u \in S} d(u) \cdot |S|
&= \frac{1}{\alpha} \sum_{u \in S} d(u) \cdot \sum_{v \in S} 1 \\
&\le \frac{1}{\alpha^2} \sum_{u \in S} \sum_{v \in S} d(u) \cdot d(v) \\
&= \frac{1}{\alpha^2} |E(S, V)|^2.
\end{align*}

Finally, they were able to relate the term $|E(S, V)|$ to $X = |E(S, \overline S)|$ using the fact that $\phi(G) = \Omega(1)$.
Consequently, by choosing $\alpha = 1/\eps$, they obtained that $\Var{I_S} \le O(\eps^2) \cdot X^2$.

However, there is no standard generalization of the notion of conductance for directed graphs.
If we look at the \emph{undirected} version of a directed graph $G$, the term $|E(S, V)|$ becomes $|E(S, S)| + |E(S, V\setminus S)| + |E(V\setminus S, S)|$.
We are only interested in approximating the directed cut value $X = |E(S, V \setminus S)|$, and because the graph is $\beta$-balanced, the best bound we have on $\overline{X} = |E(V\setminus S, S)|$ is $\overline{X} \le \beta X$.
This will result in an upper bound of $O(\frac{\beta^2}{\alpha^2} \cdot \eps^2) X^2$ on the variance of the estimator, and we will have to set $\alpha = \beta / \eps$ to make this variance small enough.
Consequently, the cut sketch will have size $\tilde O(n \beta)$ which does not have the right dependence on $\beta$.

\end{document}